\documentclass[11pt]{amsart}
\usepackage{amsmath,amsxtra,amssymb,amsthm,amsfonts,etoolbox,setspace}

\usepackage{color,hyperref}

\numberwithin{equation}{section}
\newtheorem{lemma}{Lemma}[section]
\newtheorem{prop}[lemma]{Proposition}
\newtheorem{theorem}[lemma]{Theorem}

\newtheorem{cor}[lemma]{Corollary}

\newtheorem{prob}[lemma]{Problem}
\newtheorem{rem}[lemma]{Remark}

\newtheorem{conv}[lemma]{Convention}

\newcommand{\re}{\begin{rem}\rm}
  \newcommand{\mar}{\end{rem}}

\newtheorem{defi}[lemma]{Definition}

\newcommand{\kla}{\left ( }
\newcommand{\mer}{\right ) }

\newcommand{\ee }{\mathrm{I}\!\!1}

\newcommand{\ketbra}[1]{|{#1}\rangle\langle{#1}|}

\renewcommand{\for}{\begin{eqnarray*}}

\newcommand{\mel}{\end{eqnarray*}}

\newcommand{\kl}{\pl \le \pl}
\newcommand{\gl}{\pl \ge \pl}

\newcommand{\lel}{\pl = \pl}

\newcommand{\ez}{{\mathbb E}}

\renewcommand{\L}{\mathcal{L}}

\newcommand{\nz}{{\mathbb N}}

\newcommand{\rz}{{\mathbb R}}

\newcommand{\Mz}{{\mathbb M}}

\newcommand{\cz}{{\mathbb C}}

\newcommand{\ten}{\otimes}

\DeclareMathOperator{\Lin}{Lin}

\newcommand{\p}{\hspace{.05cm}}
\newcommand{\pl}{\hspace{.1cm}}
\newcommand{\pll}{\hspace{.3cm}}

\newcommand{\hz}{\vspace{0.5cm}}

\newcommand{\qd}{\end{proof}\vspace{0.5ex}}

\newcommand{\Om}{\Omega}
\newcommand{\om}{\omega}

\newcommand{\al}{\alpha}

\newcommand{\si}{\sigma}

\newcommand{\Si}{\Sigma}

\newcommand{\La}{\Lambda}
\newcommand{\la}{\lambda}
\newcommand{\eps}{\varepsilon}

\newcommand{\id}{\iota_{\infty,2}^n}

\newcommand{\E}{{\mathcal E}}

\renewcommand{\P}{\mathcal P}

\newcommand{\pf}{\begin{proof}}

\newcommand{\be}{\left|{\atop}}

\newcommand{\xspace}{\hbox{\kern-2.5pt}}
\newcommand{\xyspace}{\hbox{\kern-1.1pt}}

\newcommand\bra[1]{\langle  #1|}
\newcommand\ket[1]{| #1\rangle}

\allowdisplaybreaks

\definecolor{LightGray}{rgb}{0.94,0.94,0.94}
\definecolor{VeryLightBlue}{rgb}{0.9,0.9,1}
\definecolor{LightBlue}{rgb}{0.8,0.8,1}
\definecolor{DarkBlue}{rgb}{0,0,0.6}
\definecolor{LightGreen}{rgb}{0.88,1,0.88}
\definecolor{MidGreen}{rgb}{0.6,1,0.6}
\definecolor{DarkGreen}{rgb}{0,0.6,0}
\definecolor{DarkGrreen}{rgb}{0,0.8,0}

\definecolor{VeryLightYellow}{rgb}{1,1,0.9}
\definecolor{LightYellow}{rgb}{1,1,0.6}
\definecolor{MidYellow}{rgb}{1,1,0.5}
\definecolor{DarkYellow}{rgb}{0.8,1,0.3}
\definecolor{VeryLightRed}{rgb}{1,0.9,0.9}
\definecolor{LightRed}{rgb}{1,0.8,0.8}
\definecolor{DarkRed}{rgb}{0.8,0.2,0}
\definecolor{DarkRedb}{rgb}{0.6,0.2,0}
\definecolor{DarkLila}{rgb}{0.8,0,1}
\definecolor{Beige}{rgb}{0.96,0.96,0.86}
\definecolor{Gold}{rgb}{1.,0.84,0.}
\definecolor{Goldb}{rgb}{0.7,0.3,0.5}
\definecolor{MyYellow}{rgb}{1.,0.84,0.8}

\newcommand{\lan}{\langle}

\def\CC{\mathbb{C}}
\def\RR{\mathbb{R}}

\def\NN{\mathbb{N}}
\def\LL{\mathbb{L}}

\def\11{\mathbb{I}}
\def\LL{\mathcal{L}}

\def\eps{\varepsilon}

\newcommand{\supp}{\mathop{\rm supp}\nolimits}

\newcommand{\tr}{\mathop{\rm Tr}\nolimits}
\renewcommand{\re}{\mathop{\rm Re}\nolimits}

\newcommand{\spec}{{\rm sp}}

\usepackage{tikz-cd}

\renewcommand{\bra}[1]{\langle#1|}
\renewcommand{\ket}[1]{|#1\rangle}

\newcommand{\cB}{{\mathcal{B}}}
\newcommand{\cD}{{\mathcal D}}
\newcommand{\cE}{{\mathcal E}}

\newcommand{\cN}{{\mathcal N}}
\newcommand{\cT}{{\mathcal T}}
\newcommand{\cH}{{\mathcal H}}
\newcommand{\cK}{{\mathcal{K}}}
\newcommand{\cJ}{{\mathcal J}}

\newcommand{\cP}{\mathcal{P}}

\newcommand{\cL}{{\mathcal L}}

\newcommand{\Id}{{\mathds{1}}}
\renewcommand{\Id}{{1}}

\def\e{\mathrm{e}}

\usepackage{graphicx}
\usepackage{setspace}
\usepackage{verbatim}
\usepackage{subfig}

\DeclareRobustCommand\openone{\leavevmode\hbox{\small1\normalsize\kern-.33em1}}

\renewcommand{\id}{\rm{id}}
\renewcommand{\be}{\begin{equation}}
	\renewcommand{\ee}{\end{equation}}
\newcommand{\bea}{\begin{eqnarray}}
	\newcommand{\eea}{\end{eqnarray}}
\newcommand{\beas}{\begin{eqnarray*}}
	\newcommand{\eeas}{\end{eqnarray*}}

\begin{document}
\title[Noncommutative change of meausure]{
Stability of logarithmic Sobolev inequalities\\ under a noncommutative change of measure}

\thanks{LG and NL acknowledge support from NSF grant DMS-1700168. NL is supported by NSF Graduate Research Fellowship Program DMS-1144245.  MJ is partially supported by NSF grants  DMS 1800872 and Raise-TAG 1839177. CR acknowledges financial support from the TUM university Foundation Fellowship and by the DFG cluster of excellence 2111 (Munich Center for Quantum Science and Technology).}

\author[M. Junge]{Marius Junge}
\address{Department of Mathematics\\
University of Illinois, Urbana, IL 61801, USA} \email[Marius
Junge]{mjunge@illinois.edu}

\author[N. LaRacuente]{Nicholas LaRacuente}
\address{Department of Physics\\
University of Illinois, Urbana, IL 61801, USA} \email[Nicholas LaRacuente]{laracue2@illinois.edu}
\author[C. Rouz\'{e}]{Cambyse Rouz\'{e}}
\address{Department of Mathematics\\
	Technische Universit\"{a}t M\"{u}nchen} \email[Cambyse Rouz\'{e}]{rouzecambyse@gmail.com}

\maketitle
\begin{abstract}
	We generalize Holley-Stroock's perturbation argument from commutative to quantum Markov semigroups. As a consequence, results on (complete) modified logarithmic Sobolev inequalities and logarithmic Sobolev inequalities for self-adjoint quantum Markov process can be used to prove estimates on the exponential convergence in relative entropy of quantum Markov systems which preserve a fixed state. This leads to estimates for the decay to equilibrium for coupled systems and to estimates for mixed state preparation times using Lindblad operators. Our techniques also apply to discrete time settings, where we show that the strong data processing inequality constant of a quantum channel can be controlled by that of a corresponding unital channel.
\end{abstract}

\section{Introduction}

Quantum information theory concerns the study of information theoretic tasks that can be
achieved using quantum systems (e.g. photons, electrons and atoms) as information carriers, with the long-term
promise that it will revolutionize our way of computing, communicating and designing new materials. However, in
realistic settings, quantum systems undergo unavoidable interactions with their environment. This gives rise to the
phenomenon of decoherence, which leads to a loss of the information initially contained in the system \cite{joos2013decoherence}. Within the
context of emerging quantum information-processing devices, gaining quantitative knowledge about the effect of
decoherence is one of the main near-term challenges for the design of methods to achieve scalable quantum fault-tolerance. Quantifying decoherence is known to be a difficult problem
in general, already for classical systems. Two facts make the situation even more challenging
in the quantum regime: (i) the non-commutativity of quantum observables, and (ii) the
potential presence of multipartite entanglement between subsystems, whose effects are notoriously
hard to characterize in precise mathematical terms.

\hz

Quantum Markov semigroups (QMS) constitute a particularly interesting class of noise for which the time interval between each use of a given
channel can be made arbitrarily small. Most recent approaches aim at quantifying decoherence for channels arising from Markov semigroups using \textit{functional inequalities} (FIs). The latter are differential versions of strong contraction properties of
various distance measures under the action of the semigroup. For instance, the Poincar\'{e} inequality provides an estimate on the spectral gap of the semigroup. Exponentially faster convergence can be achieved via the existence of a logarithmic Sobolev inequality (LSI), which implies a strong contraction of weighted $L^p$-norms under the action of the semigroup known as hypercontractivity. Similarly, the modified logarithmic Sobolev inequality (MLSI) governs the exponential convergence in relative
entropy of any initial state evolving according to the semigroup towards equilibrium. In the commutative setting, one of the key features of logarithmic Sobolev inequalities is their stability under the action of coupling with an auxiliary system. This fact implies that many such FIs can be ultimately deduced from an inequality over a two-point space. For quantum systems the stronger notion of complete (modified) logarithmic inequalities plays a similar role for studying multiplicativity properties of the semigroup.

\hz

\paragraph{\textbf{The Holley-Stroock perturbation argument}}
Our main concern in this paper is a `quantized version' of an argument by Holley and Stroock \cite{Holley1987} which allows to transfer estimates from one measure to another. For any two probability measures $\nu\ll \mu$ on $\rz^n$, their relative entropy is given by
 \[ D(\nu\|\mu)\equiv \operatorname{Ent}_\mu(f) := \int f\ln f d\mu- \int fd\mu\ln \int fd\mu \pl, \]
 where $f$ is defined as the Radon-Nikodym derivative $\frac{d\nu}{d\mu}$. Thanks to the positivity of $G(a,b)=a\log a-b\log b+b-a$ for $a,b>0$ \cite{Holley1987}, this relative entropy admits a variational characterization
 \[ \operatorname{Ent}_\mu(f) \lel \inf_{c>0} \int (f\ln f-c\ln c+c-f)\, d\mu \pl .\]
Therefore, given any other probability measure $\mu'\ll \mu$, the positivity of $G$ implies the following stability property of the relative entropy:
 \begin{equation}\label{ch1} \operatorname{Ent}_{\mu'}(f) \kl \left\|\frac{d\mu'}{d\mu}\right\|_{\infty}\, \operatorname{Ent}_{\mu}(f) \pl ,
 \end{equation}
 whenever the Radon-Nikodym derivative $\frac{d\mu'}{d\mu}$ is uniformly bounded, where $\|.\|_\infty$ refers to the $L^\infty$ norm here. A similar argument holds for the functional derivative of the relative entropy, or \textit{Fisher Information}
 \begin{align}\label{Imu}
  I_{\mu}(f) := \int L(f)\ln f d\mu
 \lel \int \frac{|\nabla(f)|^2}{f}\, d\mu \pl ,
 \end{align}
for any ``regular enough'' $f$, whenever the generator of a diffusion semigroup $(T_t=e^{-tL})_{t\ge 0}$ is given as $L(f)=-\Delta f+\nabla V.\nabla f$ with respect to the derivation $\nabla(f)=(\frac{df}{dx_1},...,\frac{df}{dx_n})$ on $\rz^n$, and for $d\mu=\e^{-V}dx$ and $V\in C^2(\mathbb{R}^n)$. Again, thanks to the positivity of  $\frac{|\nabla f|^2}{f}$, we deduce that, if $\mu\ll \mu'$
 \begin{equation}\label{ch2} I_{\mu}(f) \kl \left\|\frac{d\mu}{d\mu'}\right\|_{\infty}\, I_{\mu'}(f) \pl .
 \end{equation}
The \textit{(modified) logarithmic Sobolev inequality} (MLSI) is defined as follows: for any regular enough function $f$:
 \[ \al \operatorname{Ent}_\mu(f)\kl I_{\mu}(f) \,.\]
The largest constant $\alpha>0$ satisfying this inequality is denoted by $\al_{\mu}$ and called the \textit{modified logarithmic Sobolev constant}. Note that, by the equivalent formulation of the Fisher information in terms of differential operators (\ref{Imu}), this inequality can be merely interpreted as a property of the measure $\mu$. Hence, using the perturbation bounds previously mentioned, the Holley-Stroock perturbation bound is formulated as follows:
\begin{theorem}[Holley-Stroock \cite{Holley1987}] Let $\mu\sim\mu'$ be equivalent measures. Then
 \[ \al_{\mu}\kl \left\|\frac{d\mu}{d\mu'}\right\|_{\infty}\, \left\|\frac{d\mu'}{d\mu}\right\|_{\infty} \al_{\mu'} \pl . \]
\end{theorem}
We refer to \cite{ledoux2001logarithmic} for a wealth of interesting examples, in particular a derivation of logarithmic Sobolev inequalities at finite temperature using known estimates at infinite temperature. From a more applied angle the most impressive application of MLSI is the entropic exponential convergence of the corresponding semigroup $(\mathcal{P}_t)_{t\ge 0}$: $$\operatorname{Ent}_\mu(\mathcal{P}_t(f))\le \e^{-\alpha t}\operatorname{Ent}_\mu(f)\,.$$
The best constant working for all $f$ and $t\ge 0$ is exactly the MLSI constant $\al_{\mu}$.
\hz

Another standard procedure previously used to obtain estimates for the above entropy decay is to use an equivalent differential formulation of the notion of hypercontractivity, also known as \textit{logarithmic Sobolev inequalities} (or LSI) \cite{ledoux2001logarithmic,bakry2013analysis}. Introduced by Bobkov and Tetali \cite{Bobkov2003} for the study of Markov chains over discrete configuration spaces, the modified logarithmic Sobolev inequality turns out to be more stable. It is equivalent to LSI for classical diffusions, but provides estimates on the entropy decay of non-hypercontractive jump processes \cite{[D02]}.

\hz

Quantum functional inequalities are notoriously harder to derive than their classical analogues. For instance, only the Poincar\'{e} inequality has been shown to hold for lattice spin systems subject to the so-called \textit{heat-bath} and \textit{Davies} semigroups, under some conditions on the equilibrium Gibbs state of these evolutions \cite{Kastoryano2014,Temme2014}. The positivity of the (modified) logarithmic Sobolev constant in these settings is still unknown, and the subject of recent work \cite{bardet2019modified}. Hence, it would be very desirable to have a quantum version of Holley-Stroock's argument, because it would allow to transfer results from one reference state (say the completely mixed state) to another (say a Gibbs state at finite temperature). As we have seen, the main ingredients for the classical proof are (i) \emph{variational principle}, (ii) a good understanding of the notion of \emph{gradient}, and (iii) the  pointwise positivity of the \emph{Fisher information function} $(\nabla f,\nabla\ln f)$. Generalizing them to the quantum setting requires additional deep insight from the theory of quantum Markov semigroups and operator algebras. Such an approach is facilitated by recent developments of trace inequalities in quantum information theory.
\hz

\paragraph{\textbf{Quantum (modified) logarithmic Sobolev inequalities}} Before we state the main contribution of this paper, we first briefly recapitulate the current landscape of quantum functional inequalities: despite the existence of logarithmic Sobolev inequalities \cite{[OZ99],[TPK14]} for \textit{primitive} quantum Markov semigroups, that is for semigroups possessing a unique invariant state, it was shown in \cite{bardet2018hypercontractivity} that these inequalities cannot be derived for non-primitive semigroups. In particular, the natural notion of a logarithmic Sobolev inequality for semigroups of the form $(\cP_t\otimes\id_R)_{t\ge 0}$ given some reference system $R$, as previously introduced in \cite{[BK16a]}, is known to fail at providing entropic convergence. Just as in the classical setting, this reason motivates the introduction of a modified logarithmic Sobolev inequality for quantum Markov semigroups. The quantum MLSI was introduced by Kastoryano and Temme in \cite{[KT13]} for primitive evolutions. In \cite{BarEID17}, Bardet showed that the MLSI can also be extended to non-primitive semigroups. This led the authors of \cite{gao2018fisher} to define a notion of complete MLSI (CLSI) for the study of the convergence of the tensor product evolution of a given quantum Markov semigroup with the identity map on an arbitrarily large system. The CLSI is conjectured to hold in full generality for finite dimensional quantum Markov semigroups. Progress has been made in that direction in \cite{gao2018fisher} where it was shown that almost all finite dimensional quantum Markov semigroups that are self-adjoint with respect to the Hilbert-Schmidt inner product satisfy CLSI. These semigroups can be interpreted as evolutions occurring at infinite temperature, i.e. semigroups $(\mathcal{P}_t=e^{-t\L})_{t\ge 0}$ which are self-adjoint with respect to the trace.

\hz

For a semigroup of completely positive unital maps $\mathcal{P}_t:\mathcal{B}(\cH) \to \mathcal{B}(\cH)$ and generator $\cL:=\left.\frac{d}{dt}\right|_{t=0}\cP_t$, we denote by $\mathcal{P}_{t*}$ the adjoint with respect to the trace $\tr(\mathcal{P}_{t*}(\rho)X)=\tr(\rho \mathcal{P}_t(X))$, and $\E_*=\lim_{t\to \infty}\mathcal{P}_{t*}$ \cite{carbone2013decoherence}. In analogy with the classical setting, $\L$ is said to satisfy a \textit{modified logarithmic Sobolev inequality} if there exists a constant $\alpha>0$ such that for all density matrices $\rho$:
\begin{align*}
 D(\mathcal{P}_{t*}(\rho)\|\E_{*}(\rho))\kl \,\e^{-\alpha t} D(\rho\|\E_{*}(\rho))\pl .
 \end{align*}
The largest constant $\al$ such that this inequality holds for all $\rho$ is denoted by $\al_{\operatorname{MLSI}}(\L)$. Similarly, we denote by $\al_{\operatorname{CLSI}}(\L)$ the largest constant $\tilde{\alpha}$ such that
\begin{align*}
 D((\mathcal{P}_{t*}\otimes \id_R)(\rho)\|(\E_{*}\otimes\id_R)(\rho))\kl \,\e^{-\tilde{\alpha} t} D(\rho\|(\E_{*}\otimes\id_R)(\rho))\pl
\end{align*}
holds for all $t\ge 0$, any reference system $\cH_{R}$, and any density matrix $\rho$ on $\cH\ten \cH_R$. The advantage of the complete version is that for any two generators $\cL_1$ and $\cL_2$:
 \[ \al_{\operatorname{CLSI}}(\L_1\ten \operatorname{id}+\operatorname{id} \ten \L_2)\gl \min\{\al_{\operatorname{CLSI}}(\L_1),\al_{\operatorname{CLSI}}(\L_2)\} \pl .\]
In this article, we make another step towards proving CLSI for any finite dimensional quantum Markov semigroup by adapting the Holley-Stroock argument to the quantum setting, based  on the seminal work of Carleen and Maas \cite{Carlen20171810}. Following Carlen-Maas, the generator of a QMS satisfying the so-called \textit{detailed balance condition} (see Section \ref{sec:QMS} for more details) is given by
 \[
\cL(X)=-\sum_{j\in\cJ}\,\Big(\e^{-\omega_j/2}\,A_j^*[X,A_j]+\,\e^{\omega_j/2}[A_j,X]A_j^*\Big)\,. \]
Here, the Bohr frequencies $\omega_j\in\RR$ are determined by the additional condition $\si A_j\si^{-1}=e^{-\om_j}A_j$, for some full-rank state $\sigma$ such that $\cL_*(\sigma)=0$. Choosing these frequencies to be equal to $0$, we end up with the corresponding noncommutative \textit{heat semigroup}:
 \[
\cL_0(X)=-\sum_{j\in\cJ} \Big(A_j^*[X,A_j]+[A_j,X]A_j^*\Big)
 \lel \sum_j A_j^*A_jX+XA_jA_j^*-A_jXA_j^*-A_j^*XA_j \pl. \]

In its simplest form, our noncommutative Holley-Stroock argument can be stated as follows:
\begin{theorem}\label{main} Assume that $(\cP_t=\e^{-t\cL})_{t\ge 0}$ is a primitive quantum Markov semigroup with corresponding full-rank invariant state $\si=\sum_k \si_k |k\rangle \langle k|$ and satisfies the detailed balance condition. Then
 \[ \al_{\operatorname{CLSI}}(\L_0)\kl \frac{\max_k \si_k}{\min_k \si_k} \,\max_j e^{\om_j/2} \al_{\operatorname{CLSI}}(\L) \pl .\]
\end{theorem}

\hz
As an application of this result, we consider a primitive quantum Markov semigroups $(\mathcal{P}_{t*}=e^{-t\L_*})_{t\ge 0}$ on $\mathcal{B}(\cH)$  for finite dimensional $\cH$, which produces a certain full-rank state
$\si=\sum_k \si_k |k\rangle \langle k|$ :
 \[ \forall \rho:\,\lim_{t\to \infty} \mathcal{P}_{t*}(\rho) \lel \si  \quad \mbox{and}\quad
 \al_{\operatorname{CLSI}}(\L)>0 \pl .\]
 Our lower bound for $\al_{\operatorname{CLSI}}$ depends in an explicit way on the ratios $\frac{\si_k}{\si_l}$. On a suitable chosen inner product $(.,.)$ the derivations stabilizing $\si$ are exactly given by commutators with respect to matrix units $|k\rangle \langle j|$. In other words the density `determines' its own derivation $\delta$ and the corresponding gradient form $(\delta(f),\delta(f))$, in contrast to the above classical setting.
In our construction, we have to work with invertible densities if we want to have complete logarithmic Sobolev inequalities, and hence our results are complementary to the results in \cite{verstraete2009quantum,kraus_preparation_2008} on quantum Markov semigroups producing pure states.

\hz

\paragraph{\textbf{Outline of the paper}} In Section \ref{sec:QMS}, we recall basic aspects of the theory of quantum Markov semigroups and complete modified logarithmic Sobolev inequalities. In particular, we derive a useful form for the entropy production of a semigroup by means of noncommutative differential calculus. The essence of the quantum Holley-Stroock perturbation argument is first provided in Section \ref{sec:qHolley} where we compare a non-unital quantum Markov semigroup to a corresponding unital one. In Section \ref{HSMLSILSI}, we extend the previous argument to (i) non-primitive quantum Markov semigroups and (ii) the logarithmic Sobolev inequality. A similar argument is given in Section \ref{sec:SDPI} in order to derive strong data processing inequalities for non self-adjoint quantum channels. Sections \ref{sec:stateprep} and \ref{sec:therm} focus on applications to the dissipative preparation of mixed state and Gibbs samplers. 

\paragraph{\textbf{Notations and definitions}}
Let $(\cH,\langle .|.\rangle)$ be a finite dimensional Hilbert space of dimension $d_\cH$. We denote by $\mathcal{B}(\cH)$ the space of bounded operators on $\cH$,  by $\mathcal{B}_{\operatorname{sa}}(\cH)$ the subspace of self-adjoint operators on $\cH$, i.e. $\mathcal{B}_{\operatorname{sa}}(\cH)=\left\{X\in\mathcal{B}(\cH);\ X=X^*\right\}$, and by $\mathcal{B}_+(\cH)$ the cone of positive semidefinite operators on $\cH$, where the adjoint of an operator $Y$ is written as $Y^*$. The identity operator on $\cH$ is denoted by $\Id_\cH$, dropping the index
$\cH$ when it is unnecessary. In the case when $\cH\equiv \CC^{\ell}$, $\ell\in\NN$, we will
also use the notation $\Id$ for $\Id_{\CC^\ell}$. Similarly, we will denote by $\operatorname{id}_{\cH}$, or simply $\operatorname{id}$, resp. $\operatorname{id}_\ell$, the identity superoperator on $\mathcal{B}(\cH)$ and $\mathcal{B}(\CC^\ell)$, respectively. We denote by $\mathcal{D}(\cH)$ the set of positive semidefinite, trace one operators on $\cH$, also called \emph{density operators}, and by $\cD_+(\cH)$ the subset of full-rank density operators. In the following, we will often identify a density matrix $\rho\in\mathcal{D}(\cH)$ and the \emph{state} it defines, that is the positive linear functional $\mathcal{B}(\cH)\ni X\mapsto\tr(\rho \,X)$.

Given two positive operators $\rho,\sigma\in \mathcal{B}_+(\cH)$, the relative entropy between $\rho$ and $\sigma$ is defined as follows:
\begin{align*}
	D(\rho\|\sigma):=\left\{\begin{aligned}
		&\tr(\rho\,(\ln\rho-\ln\sigma))\,\,\,\,\supp(\rho)\subseteq\supp(\sigma)\\
		&+\infty\,\,\qquad\qquad\qquad\text{else}
	\end{aligned}\right.
\end{align*}

We recall that, given $\cN\subset \mathcal{B}(\cH)$ a finite dimensional von-Neumann subalgebra of $\mathcal{B}(\cH)$ and a full-rank state $\sigma\in\cD(\cH)$, a linear map $\cE:\mathcal{B}(\cH)\to\cN$ is called a \textit{conditional expectation} with respect to $\sigma$ of $\mathcal{B}(\cH)$ onto $\cN$ if the following conditions are satisfied:
\begin{itemize}
	\item[(i)] For all $X\in\mathcal{B}(\cH)$, $\|\cE(X)\|_\infty\le \|X\|_\infty$;
	\item[(ii)] For all $X\in \cN$, $\cE(X)=X$;
	\item[(iii)] For all $X\in\mathcal{B}(\cH)$, $\tr(\sigma\, \cE(X))=\tr(\sigma X)$.
\end{itemize}

\section{Quantum Markov semigroups and entropy decay}\label{sec:QMS}

In this section, we briefly review the notions of quantum Markov semigroups and their related noncommutative derivations on the algebra $\mathcal{B}(\cH)$ of bounded operators on a finite-dimensional Hilbert space, and explain how in this framework, the generator of a QMS should be interpreted as a noncommutative second order differential operator. We will have to recall and adapt some of the notations from the seminal papers by Carlen and Maas \cite{Carlen20171810,carlen2018non} for Lindblad generators satisfying the detailed balance condition (see also \cite{fagnola2007generators}).

\hz

\textbf{Quantum Markov semigroups and noncommutative derivations:} The basic model for the evolution of an open system in the Markovian regime is given by a quantum Markov semigroup (or QMS) $(\cP_t)_{t\ge0}$ acting on $\mathcal{B}(\cH)$. Such a semigroup is characterised by its generator, called the Lindbladian $\LL$, which is defined on $\mathcal{B}(\cH)$ by $\cL(X)={\lim}_{t\to 0}\,\frac{1}{t}\,(X-\cP_t(X))$ for all $X\in\mathcal{B}(\cH)$, so that $\cP_t=\e^{-t\cL}$\footnote{Let us note that our sign convention is opposite to the one usually used in the community of open quantum systems, but more common in abstract semigroup theory.}. The QMS is said to be \textit{primitive} if it admits a unique full-rank invariant state $\sigma$. In this paper, we exclusively study QMS that satisfy the following \textit{detailed balance condition} with respect to some full-rank invariant state $\sigma$ (also referred to as GNS-symmetry): for any $X,Y\in\mathcal{B}(\cH)$ and any $t\ge 0$:
\begin{align}\label{eq:DBC}\tag{$\sigma$-DBC}
\tr(\sigma\,X^*\cP_t(Y))=\tr(\sigma\,\cP_t(X)^*Y)\,.
\end{align}
In particular, this condition is known to be equivalent to (i) self-adjointness of the generator with respect to the so-called KMS inner product 
\begin{align}\label{def:KMS}
\langle A,B\rangle_{\sigma}:=\tr(\sigma^{\frac{1}{2}}A^*\sigma^{\frac{1}{2}}B)
\end{align}
and (ii) commutation with the modular group of $\sigma$: $\Delta_\sigma^{it}\circ \cL=\cL\circ \Delta_\sigma^{it}$ for all $t\in \RR$, where $\Delta_\sigma(X):=\sigma X\sigma^{-1}$. It was also shown in \cite{Carlen20171810} that the generator of such semigroups can take the following GKLS form (\cite{Lind,[GKS76]}): for all $X\in\mathcal{B}(\cH)$,
\begin{align}\label{eqlindblad}
\cL(X)=-\sum_{j\in\cJ}\,\Big(\e^{-\omega_j/2}\,A_j^*[X,A_j]+\,\e^{\omega_j/2}[A_j,X]A_j^*\Big)\,.
\end{align}
where the sum runs over a finite number of \textit{Lindblad operators} $\{A_j\}_{j\in\mathcal{J}}=\{A_j^*\}_{j\in\mathcal{J}}$ and $[\cdot,\cdot]$ denotes the commutator defined as $[X,Y]:=XY-YX$, $\forall X,Y\in\mathcal{B}(\cH)$, and $\omega_j\in\RR$.  Moreover, the Lindblad operators $A_j$ satisfy the following relations:
\begin{align}\label{eq}
\forall s\in\mathbb{R},\, \Delta^s_\sigma(A_j):=\sigma^s\,A_j\,\sigma^{-s}=\e^{-\omega_js}\,A_j\,\qquad \Rightarrow \qquad \delta_{A_j}(\ln\sigma)=-\omega_j A_j\,,
\end{align}
where the second identity comes from differentiability of the first one at $s=0$. Therefore, the reals $\omega_j$ can be interpreted as differences of eigenvalues of the Hamiltonian corresponding to the Gibbs state $\sigma$, also called \textit{Bohr frequencies}. It is important to note that $\L$ is the generator in the Heisenberg picture. The generator $\cP_{t*}=e^{-t\L_{*}}$ in the Schr\"{o}dinger picture is defined via
 \[ \tr(\L_{*}(\rho)X) \lel \tr(\rho \,\L(X)) \pl. \]
According to \cite[Remark 3.3]{Carlen20171810} the adjoint has the form
 \begin{align}\label{adjoint}
  \L_{\textcolor{blue}{*}}(\rho)&= -\sum_{j} \Big(e^{-\om_j/2}[A_j\rho,A_j^*]+e^{\om_j/2}[A_j^*,\rho A_j] \Big) \nonumber \\
  &= \sum_j e^{-\om_j/2}(A_j^*A_j\rho-A_j\rho A_j^*)+ e^{\om_j/2}(\rho A_jA_j^*-A_j^*\rho A_j)  \pl .
  \end{align}

The generator $\cL_0:=\sum_{j\in\cJ}\cL_{A_j}$, corresponding to taking all the Bohr frequencies to $0$, satisfies the detailed balance condition with respect to the completely mixed state $\Id/d_\cH$. Because of its analogy with the classical diffusive case, its corresponding QMS is usually called the \textit{heat semigroup}. In fact, given a Lindblad operator $A$, the generators $\cL_A:=[A^*,[A,.]]$ satisfies the following non-commutative integration by parts:
\[ \tr(X^*\cL_A(X)) \lel \tr(\delta_A(X)^*\delta_A(Y)) \pl =\tr(\cL_A(X)^*\,Y).\]
where $\delta_A(X) := [A,X]$ is a noncommutative \textit{derivation}. We may also consider $B=\frac{A+A^*}{\sqrt{2}}$ and $C=\frac{A-A^*}{\sqrt{2}i}$ and observe that
\begin{equation}\label{saver1}
\cL_A(X):=(B^2+C^2)X+X(B^2+C^2)-2(BXB+CXC)\,.
\end{equation}
has the standard form of a self-adjoint Lindbladian, with corresponding self-adjoint Lindblad operators $B$ and $C$, and in particular is $^*$-preserving.
In the GNS-symmetric case, the integration by parts formula reads as follows:
\begin{align}\label{intbypart}
\langle \cL(X),\,Y\rangle_\sigma=\sum_{j\in\cJ}\,\langle \delta_{A_j}(X),\,\delta_{A_j}(Y)\rangle_\sigma\,,
\end{align}
where the KMS inner product was defined in (\ref{def:KMS}). Because of their particular symmetry property, self-adjoint semigroups (that is w.r.t. the Hilbert-Schmidt inner product) are currently better understood than their GNS-symmetric generalizations \cite{[MHFW15],bardet2019group}. The purpose of this paper is to derive a technique to transfer estimates on the entropic rate of convergence towards equilibrium of $(\e^{-t\cL})_{t\ge 0}$ in terms of that of $(\e^{-t\cL_0})_{t\ge 0}$. The idea is to use the following commuting diagram
 \begin{equation}\label{dia}
 \begin{array}{ccc} \mathcal{B}(\cH) &\stackrel{e^{-t\cL_0}}{\to} & \mathcal{B}(\cH) \\
\downarrow_{\Gamma_{\sigma}}  & & \downarrow_{\Gamma_{\sigma}}\\
\cT_1(\cH)
 &\stackrel{e^{-t\cL}}{\to}  & \cT_1(\cH)
\end{array}
\end{equation}
where $\Gamma_{\sigma}(x)=\sigma^{1/2}x\sigma^{1/2}$ is the canonical completely positive map from the algebra $\mathcal{B}(\cH)$ to the space $ \cT_1(\cH)$ which can be interpreted as the predual $\mathcal{B}(\cH)_*$ of $\mathcal{B}(\cH)$ \cite{[OZ99],[TPK14],[KT13],[MF16]}. Indeed, we recall from \cite{Carlen20171810} that $\L$ is also self-adjoint with the KMS inner product $\langle X,Y\rangle_{\si}=\tr(\Gamma_{\si}(X^*)Y)$, and hence
 \begin{align*}
 \tr(\L_*(\Gamma_{\si}(X^*))Y) = \tr(\Gamma_{\si}(X^*)\L(Y)) \lel \langle X,\L(Y)\rangle_{\si}
 =\langle\L(X),Y\rangle_{\si} \lel \tr(\Gamma_{\si}(\L(X^*))Y) \pl
 \end{align*}
shows that indeed 
\begin{align}\label{eq:KMSsa}
\L_*(\Gamma_{\si}(X))=\Gamma_{\si}(\L(X))\,.
\end{align}

\hz

\textbf{Entropic convergence of QMS:} Under the condition of GNS-symmetry, the semigroup $(\e^{-t\cL})_{t\ge 0}$ is known to be \textit{ergodic} \cite{frigerio1982long}: there exists a \textit{conditional expectation} $\cE$ onto the fixed-point algebra $\mathcal{F}(\cL):=\{X\in\mathcal{B}(\cH):\, \cL(X)=0\}$ such that
\begin{align*}
\e^{-t\cL}\underset{t\to\infty}{\to}\cE\,.
\end{align*}
In this paper, we are interested in the exponential convergence in relative entropy of the semigroup towards its corresponding conditional expectation. The \textit{entropy production} (also known as \textit{Fisher information}) of $(\cP_t=\e^{-t\cL})_{t\ge 0}$ is defined as the opposite of the derivative of the relative entropy with respect to the invariant state: for any $\rho\in\cD(\cH)$,
\begin{align*}
\operatorname{EP}_{\cL}(\rho):=-\left.\frac{d}{dt}\right|_{t=0}\,D(\cP_{t*}(\rho)\|\cE_*(\rho))=\tr(\cL_*(\rho)(\ln\rho-\ln\sigma))\,,
\end{align*}
where the expression on the right hand side of the above equation was first proved in \cite{S78} in the primitive setting. We will also need to extend the definition of the entropy production to non-normalized states $\rho$ using the same expression as on the right-hand side of the above equation. In this paper, we are interested in the uniform exponential convergence in relative entropy of systems evolving according to a QMS towards equilibrium: more precisely, we ask the question of the existence of a positive constant $\alpha>0$ such that the following holds, for any $\rho\in\cD(\cH)$,
\begin{align*}
D(\cP_{t*}(\rho)\|\cE_*(\rho))\le \e^{-\alpha t}D(\rho\|\cE_*(\rho))\,.
\end{align*}
After differentiation at $t=0$ and using the semigroup property, this inequality is equivalent to the following \textit{modified logarithmic Sobolev constant} (MLSI) \cite{[KT13],beigi2018quantum,bardet2019modified,capel2018quantum,muller2016relative}: for any $\rho\in\cD(\cH)$:
\begin{align}\label{MLSI}\tag{MLSI}
\alpha \,D(\rho\|\cE_*(\rho))\le \operatorname{EP}_{\cL}(\rho)\,.
\end{align}
The best constant $\alpha$ achieving this bound is called the modified logarithmic Sobolev constant of the semigroup, and is denoted by $\alpha_{\operatorname{MLSI}}(\cL)$. We may also consider the complete version which requires
\begin{align}\label{CLSI}\tag{CLSI}
\alpha_{\operatorname{CLSI}}(\cL) \,D(\rho\|(\cE_*\ten \id) (\rho))\le \operatorname{EP}_{(\cL\ten \id)}(\rho)\,.
\end{align}
to hold for all $\rho \in \mathcal{B}(\cH_A\ten \cH_B)$ for any system $B$ (or even $\mathcal{B}(\cH_B)$ replaced by a finite-dimensional von Neumann algebra).
	
\hz
\textbf{Primitive semigroups:}
Our main goal is to establish MLSI and CLSI for primitive semigroups, given similar knowledge for self-adjoint semigroups. Recall that $(\P_{t}=e^{-t\L})_{t\ge 0}$ is called primitive if $\P_{t*}(\rho)=\rho$ for all $t$ implies $\rho=\si$. This is equivalent to
  \[ \L_{*}(\rho) \lel 0  \quad \Rightarrow \quad \rho=\si \pl .\]
   We recall that $\L$ in Equation \eqref{eqlindblad} is self-adjoint with respect to the inner product
$\langle A,B\rangle_{\sigma}=\tr(A^*\si^{1/2} B\si^{1/2})$. Therefore, we deduce that
  \[  0\lel
  \langle \L(X),X\rangle_{\sigma} \lel \tr(\si^{1/2}\L(X)^*\si^{1/2}X) \lel \tr(\L(X^*)\si^{1/2}X\si^{1/2}) \lel \tr(X^*\L_{*}(\si^{1/2}X\si^{1/2}))
     \]
if and only if $[A_j,X]=0$ for all $j$, by Equation (\ref{intbypart}). This implies that $\L_{*}(\si^{1/2}X\si^{1/2})=0$ if and only if $X\in \{A_j:j\in J\}'$. Let us state this for later references.

\begin{lemma}\label{ergodic} Let $\L_*$ be given by Equation \eqref{adjoint}. The following are equivalent.
\begin{enumerate}
\item[i)] $\L$ is primitive with respect to $\si$;
\item[ii)] $\{A_j:j\in J\}'\lel \cz \Id$;
\item[iii)] $\L_0 \lel \sum_j \L_{A_j}$  is ergodic, i.e.
$\L_0(X)=0$ implies $X=\la \Id$.
\end{enumerate}
\end{lemma}

\hz
\textbf{Noncommutative differential calculus via double operator integrals:}
The entropy production can be written in a different form that will be more convenient for our purpose. In order to derive it, we first need to recall some notions of non-commutative differential calculus (see \cite{daleckii1951formulas,daletskii1965integration,Birman1967,Birman1993,de2002double,DEPAGTER200428,Potapovsudo,Potapov}). Given an operator $L\in\mathcal{B}(H)$, as well as any two self-adjoint operators $X,Y\in\mathcal{B}_{sa}(\cH)$, define the operator
\begin{align*}
C^{X,Y}_A:=AY-XA\,.
\end{align*}
In particular $C^{X,X}_A:=\delta_A(X)$. Next, given a Borel function $h:\spec(X)\times \spec(Y)\to\RR$, and writing by $P_X$ and $P_Y$ the spectral measures of $X$ and $Y$, define the so-called \textit{double operator integral}
\begin{align*}
\mathcal{T}_h:=\int\,h\,L_{P_X}\,R_{P_Y}\,.
\end{align*}
where $L_{Z}$, resp.~$R_{Z}$, is the operator of left, resp.~right multiplication by $Z$. Given a differentiable function $f:\RR\to\RR$, we are exclusively interested in the restriction of the difference quotient $\tilde{f}$  associated with $f$ given by
\begin{align}\label{phiform}
\tilde{f}(x,y):=\left\{
\begin{aligned}
&\frac{f(x)-f(y)}{x-y}\,\,~~~~~~~~~~~~~~~~~~~~\text{ if }(x,y)\in\spec(X)\times \spec(Y)\\
&\frac{\partial f(x)}{\partial x}~~~~~~~~~~~~~~~~~~~~~~~~~~~~~~~~~~~~~~~~~\text{ else}\
\end{aligned}	
\right.~~.
\end{align}

\begin{theorem}[Noncommutative chain rule for differentiation, see \cite{Birman1993}]
	Given an operator $A\in\mathcal{B}(\cH)$, any two self-adjoint operators $X,Y\in\mathcal{B}_{sa}(\cH)$ and a Borel function $f:\RR\to\RR$, the following holds:
	\begin{align*}
	C_A^{f(X),f(Y)} &= \mathcal{T}_{\tilde{f}}^{X,Y}(C_A^{X,Y})\,.
	\end{align*}
\end{theorem}

With this theorem at hand, the following result can be proved:
\begin{lemma}
Assume that the QMS $(\cP_t=\e^{-t\cL})_{t\ge 0}$ satisfies \ref{eq:DBC}. Then, for any positive operator $\rho=\Gamma_\sigma(X)$,
	\begin{align}
	\operatorname{EP}_\cL(\rho)=\sum_{j\in\cJ}\,\langle \Gamma_{\sigma}(\delta_{A_j}(X)),\,[\Gamma_\sigma(X)]_{\omega_j}^{-1}(\Gamma_{\sigma}(\delta_{A_j}(X)))\rangle_{\operatorname{HS}}\,.
	\end{align}
Moreover, the same formula holds for positive $\rho\in \mathcal{B}(\cH\ten \cK)$ and $\cL$ replaced by $\cL\otimes\id_\cK$.
\end{lemma}
\begin{proof}
	By definition, for all positive $\rho\in\mathcal{B}(\cH)$, and any $\sigma\in\mathcal{F}(\cL)$, letting $X:=\Gamma_\sigma^{-1}(\rho)$ we have
	\begin{align*}
	\operatorname{EP}_{\cL}(\rho)&=\tr(\cL_*(\rho)(\ln\rho-\ln\sigma))\\
	&=\langle \cL(X),\,\ln\rho-\ln\sigma\rangle_{\sigma}\\
	&=\sum_{j\in\cJ}\,\langle\delta_{A_j}(X),\,\delta_{A_j}(\ln\rho-\ln\sigma)\rangle_\sigma\,.
	\end{align*}
Here the second line follows by Equation \ref{eq:KMSsa}, whereas the third line follows by the integration by parts formula (\ref{intbypart}). Now, due to (\ref{eq}), $\delta_{A_j}(\ln\sigma)=-\omega_j A_j$, so that, denoting $Y_j:=\rho\,\e^{-\omega_j/2}$ and $Z_j:=\rho\,\e^{\omega_j/2}$, we have
	\begin{align*}
	\delta_{A_j}(\ln\rho-\ln\sigma)&=A_j\ln(Y_j)-\ln(Z_j)A_j\\
	&=C_{A_j}^{\ln(Z_j)\,, \ln(Y_j)}\\
	&=\mathcal{T}_{\tilde{\ln}}^{Z_j,\,Y_j}(C^{Z_j,\,Y_j}_{A_j})\\
	&=\mathcal{T}_{\tilde{\ln}}^{Z_j,Y_j}(\sigma^{1/2}\delta_{A_j}(X)\sigma^{1/2})\,,
	\end{align*}
	where the last equation follows once again from Equation (\ref{eq}). Therefore, using that $\mathcal{T}_{\tilde{\ln}}^{Z_j,Y_j}:=\int_0^\infty\,(r+\e^{-\omega_j/2}L_{\Gamma_{\sigma}(X)})^{-1}(r+\e^{\omega_j/2}R_{\Gamma_{\sigma}(X)})^{-1}dr\,=: [\Gamma_{\sigma}(X)]_{\omega_j}^{-1}$ in the notations of \cite{Carlen20171810}, we end up with
\begin{align}\label{entropyprodFisher}
\operatorname{EP}_{\cL}(\rho)=\sum_{j\in\cJ}\,\langle \Gamma_{\sigma}(\delta_{A_j}(X)),\,[\Gamma_\sigma(X)]_{\omega_j}^{-1}(\Gamma_{\sigma}(\delta_{A_j}(X)))\rangle_{\operatorname{HS}}\,.
\end{align}
For $\rho\in \mathcal{B}(\cH\ten \cK)$, we observe that $\cE_*\otimes\id_{\cK}(\rho)=\sigma\otimes \tr_{\cH}(\rho)$. Since all the $A_j$'s act on the first register, the calculation above remains true.
	\end{proof}

\section{From unital to non-unital quantum Markov semigroups}\label{sec:qHolley}

We are now able to provide a quantum extension of the Holley-Stroock argument. In this section, we restrict ourselves to the primitive case and assume that $\si=\sum_k \si_k |k\rangle\lan k|$ is a positive definite density matrix of corresponding eigenbasis $\{|k\rangle\}$ of $\cH$. 

\begin{theorem}\label{theocomp}
	Let $\cL$ be the generator of a primitive, $\operatorname{GNS}$-symmetric QMS with respect to a full-rank state $\sigma$, $\cL_0$ be generator of its corresponding heat semigroup, and $\om_j$ its Bohr frequencies.
Then
	\begin{align*}
\al_{\operatorname{MLSI}}(\L_0) &\le \max_{k,l} \frac{\si_k}{\si_l} \,\max_{j} e^{|\om_j|/2}\alpha_{\operatorname{MLSI}}(\cL)\pl ,
\end{align*}
Similarly,
	\begin{align*}
\al_{\operatorname{CLSI}}(\L_0) &\le \max_{k,l} \frac{\si_k}{\si_l} \,\max_{j} e^{|\om_j|/2}\alpha_{\operatorname{CLSI}}(\cL)\pl .
\end{align*}
\end{theorem}

\begin{rem}
Using interpolation techniques, the authors of \cite{[TPK14]} showed lower bounds on the logarithmic Sobolev constant $\alpha_2$ of primitive QMS that are self-adjoint with respect to the KMS inner product. Moreover, since we further assume the detailed balance condition, our semigroups satisfy $\alpha_{\operatorname{MLSI}}(\Phi)\ge 2\alpha_2(\Phi)$, by the so-called $L_p$-regularity of Dirichlet forms proved in \cite{BarEID17}. Combining these two results, we can find that 
	\begin{align*}
	\alpha_{\operatorname{MLSI}}(\cL)\ge\frac{2\lambda(\cL)}{\ln\|\sigma^{-1}\|_\infty+2}\,,\qquad\alpha_{\operatorname{MLSI}}(\cL^{(n)})\ge \frac{2\lambda(\cL)}{\ln(d_\cH^4\,\|\sigma^{-1}\|_\infty)+11}\,,
	\end{align*} 
	where $\cL^{(n)}$ stands for the generator of the $n$-fold product of the semigroup $(\e^{-t\cL})_{t\ge 0}$, and where $\lambda(\cL)$ denotes the spectral gap of $\cL$. In the primitive setting, this means that the bounds that we derived are potentially worse than the ones provided in \cite{[TPK14]}. However, it was shown in \cite{bardet2018hypercontractivity} that the logarithmic Sobolev inequality does not hold for non-primitive QMS. This in particular means that a $\operatorname{CLSI}$ constant cannot be obtained by $L_p$-regularity. In the next section, we provide a more general result for any finite dimensional, non-primitive GNS symmetric QMS.
	\end{rem}

As in the classical case, the proof is separated in two parts: a comparison of the relative entropies, and a comparison of the entropy productions. We review these separately in the next two paragraphs.

\hz
\textbf{Comparison of relative entropies:}
We are now concerned with the left-hand side of the MLSI/CLSI. First, we need to extend the definition of the relative entropy to the case where $\rho$ and $\sigma$ are (possibly non-normalized) positive operators \cite{lindblad1974expectations}: $$D_{\operatorname{Lin}}(\rho\|\sigma):=\tr(\rho\,(\ln\rho-\ln\sigma))+\tr(\sigma)-\tr(\rho)\,,$$
where the right-hand side can be equal to infinity. As for its restriction to normalized density matrices, this relative entropy is positive. Moreover:
\begin{lemma} Lindblad's relative entropy satisfies the following properties \cite{lindblad1974expectations}:
	\begin{itemize}
		\item[(i)] \underline{Data processing inequality}:	For any positive operators $X,Y$, and any CPTP map $\Phi$, $$D_{\operatorname{Lin}}(\Phi(X)\|\Phi(Y))\le D_{\operatorname{Lin}}(X\|Y)\,.$$
		\item[(ii)] \underline{Addition under direct sums}: For any positive operators $X_1,Y_1$, resp. $X_2,Y_2$, on $\cH_1$, resp. $\cH_2$, $$D_{\operatorname{Lin}}(X_1\oplus X_2\|Y_1\oplus Y_2)=D_{\operatorname{Lin}}(X_1\|Y_1)+D_{\operatorname{Lin}}(X_2\|Y_2)\,.$$
		\item[(iii)] \underline{Normalization}: For any operators $X,Y\ge 0 $ and any constant $\lambda> 0$ $$D_{\operatorname{Lin}}(\lambda X\|\lambda Y)=\lambda\,D_{\operatorname{Lin}}(X\|Y)\,.$$
\end{itemize}
\end{lemma}
We will also need the following:
\begin{lemma}[Chain rule for $D_{\operatorname{Lin}}$]\label{chaintule}
Let $\mathcal{E}:\mathcal{B}(\cH)\to\cN$ be a conditional expectation onto a $*$-subalgebra of $\mathcal{B}(\cH)$. Then, for any $X,Y\in\mathcal{B}_+(\cH)$ such that $Y=\cE_*(Y)$, we have
\begin{align*}
D_{\operatorname{Lin}}(X\|Y)=D_{\operatorname{Lin}}(X\|\cE_*(X))+D_{\operatorname{Lin}}(\cE_*(X)\|Y)\,.
\end{align*}
	\end{lemma}
\begin{proof}
	A simple calculation gives:
	\begin{align}\label{eq:Dlinchainrule}
	D_{\operatorname{Lin}}(X\|Y)&=D_{\operatorname{Lin}}(X\|\cE_*(X))+\tr(X(\ln (\cE_*(X))-\ln Y))+\tr(Y)-\tr(\cE_*(X))\,.
	\end{align}
	Moreover, since $Y=\cE_*(Y)$, one can show that $\ln (\cE_*(X))-\ln Y\in\cN$. Indeed, we define the state $\sigma\in\cD(\cH)$ such that $\cE$ is a conditional expectation with respect to $\sigma$. This implies in particular that $\Gamma_\sigma\circ \cE=\cE_*\circ \Gamma_{\sigma}$. Moreover, one can take without loss of generality $\sigma\in\cN'$. Then,
	\begin{align*}
	\ln\cE_*(X)-\ln Y=	\ln\cE_*(X)-\ln \cE_*(Y)=\ln\Gamma_{\sigma}\circ\cE(X)-\ln\Gamma_{\sigma}\circ \cE(Y)=\ln\cE(X)-\ln\cE(Y)\in\cN\,,
	\end{align*}
	where the last identity comes from the fact that $\sigma$ commutes with $\cE(X)$ and $\cE(Y)$. Then, we can replace the traces on the right hand side of Equation (\ref{eq:Dlinchainrule}) by
	\begin{align*}
	\tr(\cE_*(X)(\ln (\cE_*(X))-\ln Y))+\tr(Y)-\tr(\cE_*(X))=D_{\operatorname{Lin}}(\cE_*(X)\|Y)\,.
	\end{align*}
	 The result follows.
	\end{proof}

\begin{prop}[Noncommutative change of measure argument]\label{propequivrelent} Let $\cK$ be an additional Hilbert space,
$\E_0\lel \,\tr_{\cH}\otimes \frac{\Id_\cH}{d_\cH}$ a conditional expectation onto $ \mathcal{B}(\cK)\otimes \Id_\cH$ and
 \[ \E_*(\rho)\lel d_\cH\,d_\cK\,\Gamma_{\frac{\Id_{\cK}}{d_\cK}\otimes \si}\circ\E_0(\rho) \]
the conditional expectation on the space of densities. Then for all $X\ge0$:
 \[  D_{\operatorname{Lin}}(\Gamma_{\Id_{\cK}\otimes \si}(X)\|\,\E_*\circ\Gamma_{\Id_\cK\otimes\sigma}(X)) \kl \max_k \{\si_k\} \,D_{\operatorname{Lin}}(X\|\, \E_{0}(X)) \pl . \]
\end{prop}

\begin{proof}
Since the following inequality holds by Lemma \ref{chaintule}: $$D_{\operatorname{Lin}}(d_\cH\Gamma_{\Id_{\cK}\otimes \si}(X)\|\,d_\cH\E_*\circ\Gamma_{\Id_\cK\otimes\sigma}(X))\le D_{\operatorname{Lin}}(d_\cH\Gamma_{\Id_\cK\otimes \sigma}(X)\|\cE_*(X))\,,$$
it is enough to prove that
$$
D_{\operatorname{Lin}}(d_\cH\Gamma_{\Id_\cK\otimes \sigma}(X)\|\cE_*(X))\le d_\cH\, \max_k \{\si_k\} \,D_{\operatorname{Lin}}(X\|\, \E_{0}(X)) \,.
$$
Define the map $\Phi(X)=\Lambda^{-1}\,{\Gamma_{\Id_\cK\otimes \sigma}(X)}$, where $\Lambda:=\max_k \si_k$. This map is completely positive, trace non-increasing. We may also define
 \[ \Psi(X) := \kla \begin{array}{cc} \Phi(X)& 0\\
                                    0& \tr((1-\frac{\Id_\cK\otimes \si}{\La})X)\end{array}  \mer \pl .\]
Then $\Psi$ is trace-preserving and hence, see \cite{lindblad1974expectations}, we know that
 \[ D_{\operatorname{Lin}}(\Psi(X)\|\Psi(\E_{0}(X)) \kl D_{\operatorname{Lin}}(X\|\E_{0}(X)) \pl .\]
Since $D_{\operatorname{Lin}}(\pl\|\pl)$ is positive, we deduce from the diagonal output of $\Psi$ that indeed,
\begin{align*}
 D_{\operatorname{Lin}}(\Lambda^{-1}\,\Gamma_{\Id_{\cK}\otimes\sigma}(X) \|\Lambda^{-1}\,\Gamma_{\Id_\cK\otimes\sigma}\circ \E_{0}(X))
 &=  D_{\operatorname{Lin}}(\Phi(X)\|\Phi(\E_{0}(X)) \kl
  D_{\operatorname{Lin}}(\Psi(X)\|\Psi(\E_{0}(X)) \pl .
 \end{align*}
 By definition, $\E_{*}=d_\cH\,\Gamma_{\Id_\cK\otimes\sigma}\circ \E_0$. Therefore,
	\begin{align*}
	D_{\operatorname{Lin}}(d_\cH\Gamma_{\Id_\cK\otimes\sigma}(X)\|\E_{*}(X))
	&=D_{\operatorname{Lin}}(d_\cH\Gamma_{\Id_\cK\otimes\sigma}(X)\|d_\cH\Gamma_{\Id_\cK\otimes\sigma}\circ\E_{0}(X))\\
	&=d_\cH\,\Lambda\,D_{\operatorname{Lin}}(\Lambda^{-1}\,\Gamma_{\Id_\cK\otimes \sigma}(X)\|  \Lambda^{-1}\,\Gamma_{\Id_\cK\otimes \sigma}\circ \E_{0}(X))\\
	&\le d_\cH\,\Lambda\,D_{\operatorname{Lin}}(X\|  \E_{0}(X) )\\
	&=d_\cH\,\max_k \si_k \pl D_{\operatorname{Lin}}(X\| \E_{0}(X) )  \pl . \qedhere
	\end{align*}
\qd

\textbf{Comparison of entropy productions:} We are now interested in controlling the entropy production of $\cL_0$ in terms of that of $\cL$. Extending the expression derived in Equation \ref{entropyprodFisher} for the entropy production to non-normalized states, we find for any positive operator $X\in\mathcal{B}(\cK\otimes\cH)$:
\begin{align*}
\operatorname{EP}_{\cL}(\Gamma_{\Id_\cK\otimes\sigma}(X))&=\sum_{j\in\cJ}\,\langle \Gamma_{\Id_{\cK}\otimes \sigma}(\delta_{A_j}(X)),\,[\Gamma_{\Id_\cK\otimes \sigma}(X)]_{\omega_j}^{-1}(\Gamma_{\Id_{\cK}\otimes\sigma}(\delta_{A_j}(X)))\rangle_{\operatorname{HS}}\\
&\ge \,\inf_j\,\e^{-|\omega_j|/2}\,\sum_{j\in\cJ}\,\langle \Gamma_{\Id_\cK\otimes \sigma}(\delta_{A_j}(X)),\,[\Gamma_{\Id_{\cK}\otimes\sigma}(X)]_{0}^{-1}\,(\Gamma_{\Id_\cK\otimes\sigma}(\delta_{A_j}(X)))\rangle_{\operatorname{HS}}
\end{align*}
where we denoted $\delta_{A_i}=\Id_\cK\otimes \delta_{A_i}$ by slight abuse of notations, and where we used that, by definition $[\rho]_{\omega_j}\le \max_j\,\e^{|\omega_j|/2}\,[\rho]_0$ as self-adjoint operators in $\langle.,.\rangle_{\operatorname{HS}}$. Moreover, we need the following direct extension of Theorem 5 of \cite{hiai2012quasi} to the case of trace non-increasing maps:
\begin{prop}\label{lieb}
	Let $\Phi:\mathcal{B}(\cH)\to\mathcal{B}(\cH)$ be a completely-positive, trace non-increasing map, then the following holds for any $A>0$ and any $X\in\mathcal{B}(\cH)$:
	\begin{align*}
	\langle \Phi(X),\,[\Phi(A)]^{-1}_0(\Phi(X))\rangle_{\operatorname{HS}}\,\le \,	\langle X,\,[A]^{-1}_0(X)\rangle_{\operatorname{HS}}\,.
	\end{align*}
\end{prop}

Choose this time $\Phi'(X):=(\Lambda')^{-1}\,{\Gamma_{\Id_\cK\otimes\sigma}^{-1}(X)}$, where $\Lambda':= \max_k \{\si_k^{-1}\}$. This map is completely positive and trace non-increasing.
Then
 \begin{align*}
 \operatorname{EP}_{\cL}(\Gamma_{\Id_\cK\otimes\sigma}(X))&\ge \, \min_j\,\e^{-|\omega_j|/2}\,\sum_{j\in\cJ}\,\langle \Gamma_{\Id_\cK\otimes \sigma}(\delta_{A_j}(X)),\,[\Gamma_{\Id_\cK\otimes \sigma}(X)]_{0}^{-1}\,(\Gamma_{\Id_\cK\otimes \sigma}(\delta_{A_j}(X)))\rangle_{\operatorname{HS}}\\
 &\ge  \, \min_j\,\e^{-|\omega_j|/2}\,\sum_{j\in\cJ}\,\langle (\Lambda')^{-1}\,\delta_{A_j}(X),\,[(\Lambda')^{-1}\,X]_{0}^{-1}\,((\Lambda')^{-1} \,\delta_{A_j}(X))\rangle_{\operatorname{HS}}\\
 &=\,(\Lambda')^{-1}\,\min_j\,\e^{-|\omega_j|/2}\,\sum_{j\in\cJ}\,\langle \delta_{A_j}(X),\,[X]_{0}^{-1}\,( \delta_{A_j}(X))\rangle_{\operatorname{HS}}\\
 &\ge (\La')^{-1} \min_j\,\e^{-|\omega_j|/2}
 \,\operatorname{EP}_{\L_0}(X) \pl .
 \end{align*}

We have proved the following:
\begin{prop}\label{EPequiv}
	Let $\cL$ be the generator of a QMS that is self-adjoint with respect to the GNS inner product associated to a full-rank state $\sigma=\sum_k \si_k |k\rangle\lan k|$. Then, the following comparison of entropy productions holds: for any positive operator $X\in\mathcal{B}(\cK\otimes\cH)$,
	\begin{align*}
	\operatorname{EP}_{\cL}(\Gamma_{\Id_\cK\otimes \si}(X))\ge \min_k \{\si_k\} \min_j \{e^{-|\om_j|/2}\} \,\operatorname{EP}_{\L_0}(X)\,.
	\end{align*}
\end{prop}

Combining Propositions \ref{propequivrelent} and \ref{EPequiv}, we are now ready to prove the main result of this section, namely Theorem \ref{theocomp}:

\begin{proof}[Proof of Theorem \ref{theocomp}]
First notice the following: for any $X\in\mathcal{B}_+(\cH\ten \cK)$ and thanks to the homogeneity of the entropy production and relative entropy
	\begin{align*}
	\al_{\operatorname{CLSI}}(\cL_0) \,D_{\operatorname{Lin}}(X\|(\E_{0}\otimes\id_\cK)(X) ) \kl \operatorname{EP}_{\L_0}(X)
\pl .
\end{align*}
The result then comes directly from the following chain of inequalities:
for any $\rho\lel \Gamma_{\Id_{\cK}\otimes \si}(X)\in\cD(\cH)$:
	\begin{align*}
	\al_{\operatorname{CLSI}}(\L_0) \,D(\rho\|(\E_{*}\otimes\id_\cK)(\rho))
	&\le \,\max_k \si_k\,  \al_{\operatorname{CLSI}}(\L_0)
  D_{\operatorname{Lin}}(X\| (\E_{0}\otimes\id_\cK)(X)) \\
	&\le
\max_{k} \sigma_k \operatorname{EP}_{\L_0\otimes\id_\cK}(X)  \\
&\le \max_{k} \sigma_k  \,\max_k \sigma_k^{-1} \max e^{|\om_j|/2}
\operatorname{EP}_{\L\otimes\id_\cK}(\rho) \pl.
\end{align*}
	The first inequality follows from Proposition \ref{propequivrelent}, the second one by the definition of the CLSI constant of $\cL_0$, and the last one by Proposition \ref{EPequiv}.
\end{proof}

\section{A non-primitive Holley-Stroock perturbation argument}\label{HSMLSILSI}
In this section, we extend Theorem \ref{theocomp} in essentially two directions: first, we do not assume that the reference semigroup $(\e^{-t\cL_0})_{t\ge 0}$ is the heat semigroup. Second, we relax the condition of primitivity for $(\mathcal{P}_t)_{t\ge 0}$. This situation will in particular extend the argument for CLSI. We will be interested in both MLSI and LSI inequalities.
\subsection{Perturbing the modified logarithmic Sobolev inequality}\label{sec:HSnp}

 More precisely, we want to compare the MLSI constants of the following two generators safistying the detailed balance condition with respect to two different, though commuting states:
\begin{align}\label{eqlindbladsigma}
\cL_{\sigma}(X)=-\sum_{j\in\cJ}\,\Big(\e^{-\omega_j/2}\,A_j^*[X,A_j]+\,\e^{\omega_j/2}[A_j,X]A_j^*\Big)\,.
\end{align}
and
\begin{align}\label{eqlindbladsigmaprim}
\cL_{\si'}(X)=-\sum_{j\in\cJ}\,\Big(\e^{-\nu_j/2}\,A_j^*[X,A_j]+\,\e^{\nu_j/2}[A_j,X]A_j^*\Big)\,.
\end{align}
Remark that these generators are given by the same Lindblad operators $\{A_j\}_{j\in\mathcal{J}}$, and only differ at the level of their Bohr frequencies. This implies in particular that they share the same fixed-point algebra $\mathcal{F}=\{K_j:\,j\in\mathcal{J}\}'$, which we decompose into matrix blocks:  
\begin{align*}
\mathcal{F}=\bigoplus_{i\in
	\mathcal{I}}\,\mathcal{B}(\cH_i)\otimes\Id_{\cK_i}\,.
\end{align*}
By the detailed balance condition, the operators $A_j$ are eigenvectors of the modular groups corresponding to two full-rank invariant states $\si$, resp. $\si'$, which can without loss of generality be taken as follows: given two families of full-rank states $\{\tau_i\}_{i\in\mathcal{I}}$ and $\{\tau_i'\}_{i\in\mathcal{I}}$,
\begin{align*}
\si:=\sum_{i\in\mathcal{I}}\,\frac{d_{\cK_i}}{d_\cH}\Id_{\cH_i}\otimes\tau_i\,,\qquad \si':=\sum_{i\in\mathcal{I}}\,\frac{d_{\cK_i}}{d_\cH}\Id_{\cH_i}\otimes\tau_i'\,,
\end{align*}
 In particular,
\[ \Delta_{\si}(A_j)  \lel e^{-\om_j} A_j ,\qquad
\Delta_{\si'}(A_j) \lel e^{-\nu_j} A_j \pl . \]
This implies in particular that the states $\tau_i$ and $\tau_i'$ commute so that:
\[\tau_i=\sum_{k}\,\lambda_k^{(i)}\,P_k^{(i)}\,,\qquad \tau_i'=\sum_{k}\,\lambda_k^{(i)'}\,P_k^{(i)}\,.\]
Our general perturbation theorem follows:
\begin{theorem}[Non-primitive Holley-Stroock for MLSI]\label{theononpr}
	With the above notations, the following holds:
		\begin{align*}
	\al_{\operatorname{MLSI}}(\L_{\si'}) &\le  \max_{i\in\mathcal{I},k} \,\frac{\la_k^{(i)} }{\la_k^{(i)'}}
	\, \max_{i\in\mathcal{I},k} \,\frac{\la_k^{(i)'} }{\la_k^{(i)}}
	\,\max_{j\in\mathcal{J}} e^{|\om_j-\nu_j|/2}\alpha_{\operatorname{MLSI}}(\cL_\si)\pl .
	\end{align*}
	Similarly,
	\begin{align*}
	\al_{\operatorname{CLSI}}(\L_{\si'}) &\le  \max_{i\in\mathcal{I},k} \,\frac{\la_k^{(i)} }{\la_k^{(i)'}}
	\,\,\max_{i\in\mathcal{I},k} \,\frac{\la_k^{(i)'} }{\la_k^{(i)}}
	\,\max_{j\in\mathcal{J}} e^{|\om_j-\nu_j|/2}\alpha_{\operatorname{CLSI}}(\cL_\si)\pl .
	\end{align*}
	\end{theorem}
The proof of the theorem follows the same lines as for that of Theorem \ref{theocomp}. We compare relative entropies and entropy productions separately:

\hz
\textbf{Comparison of relative entropies:} 
\begin{prop}[Change of measure]\label{propequivrelent1} Denote $\cE_{\si}:=\lim_{t\to \infty}\e^{-t\cL_{\si}}$ and $\cE_{\si'}:=\lim_{t\to \infty}\e^{-t\cL_{\si'}}$. Then, for all $X\ge0$:	
 \[ D_{\Lin}(\Gamma_{\si}(X)||\cE_{\si*}\circ\Gamma_{\si}(X))\kl  \max_{i\in\mathcal{I},k} \,\frac{\la_k^{(i)} }{\la_k^{(i)'}}
\pl D_{\Lin}(\Gamma_{\si'}(X)||\cE_{\si'*}\circ\Gamma_{\si'}(X)) \pl .\]	
\end{prop}

\begin{proof}
	The conditional expectations $\cE_\si$ and $\cE_{\si'}$ take the following form: 
	\begin{align*}
	\cE_{\si*}= \sum_{i}\,\tr_{\cK_i}(P_i\,.\,P_i)\otimes \tau_i\,,\qquad 	\cE_{\si'*}= \sum_{i}\,\tr_{\cK_i}(P_i\,.\,P_i)\otimes \tau_i'\,,
	\end{align*}
	where for each $i$, $P_i$ is the projection onto the block $i$ in the decomposition of $\mathcal{F}$. This implies in particular the following relation: $\Gamma_{\si'}^{-1}\circ\cE_{\si'*}=\Gamma_{\si}^{-1}\circ\cE_{\si*}$. Next, consider the completely positive map $\Phi:= \frac{1}{r}\Gamma_\sigma\circ\Gamma_{\sigma'}^{-1} $, with $r:=\max_{i\in\mathcal{I},k}{\la^{(i)}_k}/{\la^{(i)'}_{k}}$. One can readily verify that $\Phi$ is trace non-increasing. Moreover, by Lemma \ref{chaintule}:
	\begin{align*}
	D_{\operatorname{Lin}}(\Gamma_\sigma(X)\|\cE_{\si*}\circ\Gamma_\si(X))&=D_{\operatorname{Lin}}(r\,\Phi\circ\Gamma_{\si'}(X)\|r\,\cE_{\si*}(\Phi\circ\Gamma_{\si'}(X)))\\
	&\le D_{\operatorname{Lin}}(r\,\Phi\circ\Gamma_{\si'}(X)\|\,\cE_{\si*}(\Gamma_{\si'}(X)))\\
	&=D_{\operatorname{Lin}}(r\,\Phi\circ \Gamma_{\si'}(X)\|\,r\,\Phi\circ \cE_{\si'*}\circ\Gamma_{\si'}(X))\,.
	\end{align*}
	Next, by homogeneity and data processing inequality after proper normalization of the channel as in the proof of Theorem \ref{propequivrelent}, we find that
	\begin{align*}
	D_{\operatorname{Lin}}(\Gamma_\sigma(X)\|\cE_{\si*}\circ\Gamma_\si(X))\le r\,D_{\operatorname{Lin}}( \Gamma_{\si'}(X)\| \cE_{\si'*}\circ\Gamma_{\si'}(X))\,,
	\end{align*}
	which is what needed to be proved.
	\qd

	\hz
\textbf{Comparison of entropy productions:}	We are now interested in comparing the entropy productions of $\cL_\si$ and $\cL_{\si'}$.

\begin{prop}\label{propEPcontrol}In the above notations, for any $X\ge 0$:
	\[\operatorname{EP}_{\cL_{\si'}}(\Gamma_{\si'}(X))
	\kl \max_{i\in\mathcal{I},k} \,\frac{\la_k^{(i)'} }{\la_k^{(i)}}\pl \,\max_j e^{|\om_j-\nu_j|/2}
	\operatorname{EP}_{\cL_{\si}}(\Gamma_{\si}(X)) \pl .\]
\end{prop}

	\begin{proof}
	Using the expression derived in \ref{entropyprodFisher} for the entropy production for non-normalized states, we find for any positive operator $X\in\mathcal{B}(\cH)$:
\begin{align*}
\operatorname{EP}_{\cL_\si}(\Gamma_{\sigma}(X))&=\sum_{j\in\cJ}\,\langle \Gamma_{\sigma}(\delta_{A_j}(X)),\,[\Gamma_{\sigma}(X)]_{\omega_j}^{-1}(\Gamma_{\sigma}(\delta_{A_j}(X)))\rangle_{\operatorname{HS}}\\
&\ge \,\min_j\,\e^{-|\omega_j-\nu_j|/2}\,\sum_{j\in\cJ}\,\langle \Gamma_{\sigma}(\delta_{A_j}(X)),\,[\Gamma_{\sigma}(X)]_{\nu_j}^{-1}\,(\Gamma_{\sigma}(\delta_{A_j}(X)))\rangle_{\operatorname{HS}}
\end{align*}
where we used that, by definition $[\rho]_{\omega_j}\le \max_j\,\e^{|\omega_j-\nu_j|/2}\,[\rho]_{\nu_j}$ as self-adjoint operators in $\langle.,.\rangle_{\operatorname{HS}}$. Next, we observe that $\Phi(X)=\frac{1}{R}\Gamma_{\si'}\circ\Gamma_{\si}^{-1}(X))$ is a completely positive trace non-increasing map for $R:= \max_{i\in\mathcal{I},k} \,{\la_k^{(i)'} }/{\la_k^{(i)}}$. It is easy to `complete' $\Phi$ to a trace preserving completely positive map and hence, we deduce from Theorem 5 of \cite{hiai2012quasi} that, given $Y_j:=\Gamma_{\si}(X)\,\e^{-\nu_j/2}$ and $Z_j:=\Gamma_{\si}(X)\,\e^{\nu_j/2}$, and  $Y_j':=\Gamma_{\si'}(X)\,\e^{-\nu_j/2}$ and $Z_j':=\Gamma_{\si'}(X)\,\e^{\nu_j/2}$:
	\begin{align*}
	R^{-1}
	\operatorname{EP}_{\cL_{\si'}}(\Gamma_{\si'}(X))
	&= R^{-1}\sum_{j\in J}   \langle \Gamma_{\si'}(\delta_{A_j}(X)),[\Gamma_{\si'}(X)]_{\nu_j}^{-1}(\Gamma_{\si'}(\delta_{A_j}(X)))\rangle_{\operatorname{HS}}\\
	&= R^{-1}\sum_{j\in J}   \langle \Gamma_{\si'}(\delta_{A_j}(X)), \mathcal{T}_{\tilde{\ln}}^{Y_j',Z_j'}\circ
	\Gamma_{\si'}(\delta_{A_j}(X))\rangle_{\operatorname{HS}}\\
	&= \sum_{j\in J}   \langle \Phi\circ\Gamma_{\si}(\delta_{A_j}(X)), \mathcal{T}_{\tilde{\ln}}^{\,\Phi(Y_j),\Phi(Z_j)}\circ
	\Phi\circ\Gamma_{\si}(\delta_{A_j}(X))\rangle_{\operatorname{HS}}\\
	&\le \sum_{j\in J}   \langle \Gamma_{\si}(\delta_{A_j}(X)), \mathcal{T}_{\tilde{\ln}}^{\,Y_j,Z_j}\circ
	\Gamma_{\si}(\delta_{A_j}(X))\rangle_{\operatorname{HS}}\pl \\
	&=\sum_{j\in J}   \langle \Gamma_{\si}(\delta_{A_j}(X)), [\Gamma_{\sigma}(X)]^{-1}_{\nu_j}
	\Gamma_{\si}(\delta_{A_j}(X)))\rangle_{\operatorname{HS}}\pl\,.
	\end{align*}
The assertion follows. 
	\qd

\begin{proof}[Proof of Theorem \ref{theononpr}] Combine the previous two propositions as in the proof of Theorem \ref{theocomp}.\qd

\subsection{Perturbing the logarithmic Sobolev inequality}
The above Holley-Stroock argument can be easily adapted to the setting of the logarithmic Sobolev inequality. Such an inequality was shown to provide similar decoherence times as the MLSI in the primitive \cite{[OZ99],[TPK14],[KT13]} and non-primitive \cite{bardet2018hypercontractivity} settings. 

Recall that a faithful quantum Markov semigroup $(\cP_t:=\e^{-t\cL})_{t\ge 0}$, of corresponding conditional expectation $\cE$ towards its fixed-point algebra and full-rank invariant state $\sigma$, satisfies a \textit{weak logarithmic Sobolev inequality} (LSI) with constants $c>0$ and $d\ge 0$ if the following holds: for any positive definite state $\rho$,
\begin{align}\tag{LSI($c,d$)}\label{LSI}
D(\rho\|\cE_*(\rho))\le \,c\,\cE_{\cL}(\si^{-\frac{1}{4}}\rho^{\frac{1}{2}}\si^{-\frac{1}{4}})+d\,\|\si^{-\frac{1}{4}}\rho^{\frac{1}{2}}\si^{-\frac{1}{4}}\|_{\mathbb{L}_2(\sigma)}^2\,.
\end{align}
Here, the \textit{Dirichet form} $\cE_{\cL}$ is defined as $\cE_\cL(X):=\langle \cL(X),\,X\rangle_\si$, whereas the non-commutative $\mathbb{L}_p$ norms are defined as 
\begin{align*}
\|X\|_{\mathbb{L}_p(\si)}:=\Big(  \tr|\Gamma_\si^{\frac{1}{p}}(X)|^p\,\Big)^{\frac{1}{p}}\,.
\end{align*}
\begin{theorem}
	Let $\cL_\si$ and $\cL_{\si'}$ be defined as in Section \ref{sec:HSnp}. Assume that $\cL_{\si'}$ satisfies $\operatorname{LSI}$ with constants $(c',d')$. Then, $\cL_\si$ satisfies $\operatorname{LSI}$ with constants $(c,d)$ such that
	\begin{align*}
	c\le\, \max_{i\in\mathcal{I},k} \,\frac{\la_k^{(i)} }{\la_k^{(i)'}}
	\, \max_{i\in\mathcal{I},k} \,\frac{\la_k^{(i)'} }{\la_k^{(i)}}
	\,\max_{j\in\mathcal{J}} e^{|\om_j-\nu_j|/2}\,c'\,,\qquad d\le \max_{i\in\mathcal{I},k} \,\frac{\la_k^{(i)} }{\la_k^{(i)'}}
	\, \max_{i\in\mathcal{I},k} \,\frac{\la_k^{(i)'} }{\la_k^{(i)}}d'\,.
	\end{align*}
	\end{theorem}

\begin{proof}
 Given $X\ge 0$, the entropic term on the left-hand side of \ref{LSI} is taken care of the exact same way as in Proposition \ref{propequivrelent1}: 
\[D_{\Lin}(\Gamma_{\si}(X)||\cE_{\si*}\circ\Gamma_{\si}(X))\kl  \max_{i\in\mathcal{I},k} \,\frac{\la_k^{(i)} }{\la_k^{(i)'}}
\pl D_{\Lin}(\Gamma_{\si'}(X)||\cE_{\si'*}\circ\Gamma_{\si'}(X)) \pl .\]	
	Next, by assumption and homogeneity of the LSI, we have that
	\begin{align*}
	D_{\Lin}(\Gamma_{\si'}(X)||\cE_{\si'*}\circ\Gamma_{\si'}(X))\le c'\,\cE_{\cL_{\si'}}(\si^{'-\frac{1}{4}}\Gamma_{\si'}(X)^{\frac{1}{2}}\si^{'-\frac{1}{4}})+d'\,\|\si^{'-\frac{1}{4}}\Gamma_{\si'}(X)^{\frac{1}{2}}\si^{'-\frac{1}{4}}\|_{\mathbb{L}_2(\si')}^2\,.
	\end{align*}
First, notice that
	\begin{align*}
	\|\si^{'\frac{-1}{4}}\Gamma_{\si'}(X)^{\frac{1}{2}}\si^{'\frac{-1}{4}}\|_{\mathbb{L}_2(\si')}^2=\tr(\si' X)\le  \max_{i\in\mathcal{I},k} \,\frac{\la_k^{(i)'} }{\la_k^{(i)}}\tr(\si X)= \max_{i\in\mathcal{I},k} \,\frac{\la_k^{(i)'} }{\la_k^{(i)}}\|\si^{\frac{-1}{4}}\Gamma_{\si}(X)^{\frac{1}{2}}\si^{\frac{-1}{4}}\|_{\mathbb{L}_2(\si)}^2\,,
	\end{align*}
	which directly gives the expected upper bound on $d$. For $c$, we control the Dirichlet form in a way that is completely analogous to what we did for the entropy production in the proof of Proposition \ref{propEPcontrol}: we have 
	\begin{align*}
	\cE_{\cL_{\si'}}(\si^{'-\frac{1}{4}}\Gamma_{\si'}(X)^{\frac{1}{2}}\si^{'-\frac{1}{4}})&=\sum_{j\in\mathcal{J}}\,\langle C_{A_j}^{Y_j',\,Z_j'},\,\mathcal{T}_{\tilde{f}_{1/2}^2}^{Y_j',\,Z_j'}(C_{A_j}^{Y_j,\,Z_j'})\rangle_{\operatorname{HS}}\\
	&=\sum_{j\in\mathcal{J}}\,\langle \Gamma_{\si'}(\delta_{A_j}(X)),\,\mathcal{T}_{\tilde{f}_{1/2}^2}^{Y_j',\,Z_j'}(\Gamma_{\si'}(\delta_{A_j}(X)))\rangle_{\operatorname{HS}}\,,
	\end{align*}
	where $Y_j':=\e^{\frac{\nu_j}{2}}\Gamma_{\si'}(X)$, $Z_j':=\e^{-\frac{\nu_j}{2}}\Gamma_{\si'}(X)$ and $f_{1/2}:x\mapsto x^{1/2}$. We conclude by first noticing that $$\mathcal{T}_{\tilde{f}_{1/2}^2}^{\e^{\frac{\nu_j}{2}}\tilde{X},\e^{-\frac{\nu_j}{2}}\tilde{X}}\le \max_{j\in\mathcal{J}}\e^{|\nu_j-\omega_j|/2} \,\,\mathcal{T}_{\tilde{f}_{1/2}^2}^{\e^{\frac{\om_j}{2}}\tilde{X},\e^{-\frac{\om_j}{2}}\tilde{X}}$$
	and by invoking Theorem 5 of \cite{hiai2012quasi} in the same way as we did in Proposition \ref{propEPcontrol}.
	\qd

\section{Strong data processing inequality}\label{sec:SDPI}

 In this section, we adapt the proof of Section \ref{sec:qHolley} to the discrete time setting. Let $\Phi_*:\mathcal{T}_1(\cH)\to \mathcal{T}_1(\cH)$ be a quantum channel. Assume that $\Phi_*$ has a full-rank invariant state $\sigma$, and that the following detailed balance condition holds for its dual map $\Phi$: for all $X,Y\in\cB(\cH)$:
 \begin{align}\label{discretetimeDBC}
 \tr(\sigma \,\Phi(X)^*Y)=\tr(\sigma X^*\Phi(Y))\,.
 \end{align}
 It was shown in \cite{beigi2018quantum} that $\Phi$ has the following Kraus decomposition $ \Phi(X):=\sum_{j\in\mathcal{J}}\lambda_j\,K_jXK_j^*$ for some normalization constants $\lambda_j>0$, and where the Kraus operators $\{K_j\}=\{K_j^*\}$ satisfy $\sigma K_j=\e^{-\omega_j}\,K_j\sigma$, $\omega_j\in\mathbb{R}$. The normalization constants are defined in such a way that the map $\Phi_0$ defined by 
 \begin{align}\label{phi0}
 \Phi_0(X)=\sum_{j\in \mathcal{J}}K_j^*XK_j
 \end{align}
  is unital. It is easy to see that the choice $\lambda_j=\e^{\omega_j}$ works. Therefore, we have 
 \begin{align}
 \Phi(X)=\sum_{j\in\mathcal{J}}\e^{\omega_j}K_jXK_j^*\,.
 \end{align}
Now, in complete analogy with the continuous time setting, there exists a conditional expectation, call it $\mathcal{E}$, onto the fixed point algebra $\mathcal{F}(\Phi):=\{X\in\mathcal{B}(\cH):\Phi(X)=X\}$, such that 
\begin{align*}
 \Phi^n\underset{n\to\infty}{\to}\mathcal{E}\,.
 \end{align*}
In this section, we are interested in estimating the \textit{strong data processing inequality} (SDPI) constant $c(\Phi)$, defined as the largest constant $c$ such that. for any $\rho\ge 0$:
\begin{align*}
D_{\operatorname{Lin}}(\Phi_*(\rho)\|\mathcal{E}_*(\rho))\le cD_{\operatorname{Lin}}(\rho\|\mathcal{E}_*(\rho))\,.
\end{align*} 
By the data processing inequality, and since $\Phi_*\circ\mathcal{E}_*=\mathcal{E}_*$, $c\le 1$ necessarily. Moreover, when $\Phi_*$ is unital, the constant $c$ can be estimating in terms of the logarithmic Sobolev constant of a related quantum Markov semigroup, see \cite{[MHFW15]}. Now, a direct adaptation of the Holley-Stroock argument of Section \ref{sec:qHolley} allows us to reduce to this setting. For sake of simplicity, we state our result in the primitive case so that $\sigma$ is the unique invariant state of $\Phi$:
\begin{prop}\label{prop:sdpi}
Let $\Phi_*$ be a primitive quantum channel of unique invariant state $\sigma$ and satisfying Equation $\operatorname{(}$\ref{discretetimeDBC}$\operatorname{)}$, and let $\Phi_0$ be the corresponding unital channel defined as in Equation $\operatorname{(}$\ref{phi0}$\operatorname{)}$. Then, 
\begin{align*}
 c(\Phi)\le 
\operatorname{min}\{1,\, \|\sigma\|_\infty\,\|\sigma^{-1}\|_\infty\,c(\Phi_0)\}\,.
\end{align*}
	\end{prop}

\begin{proof}
First of all, we notice that $\Phi_*\circ \Gamma_\sigma=\Gamma_\sigma\circ\Phi_0$. Indeed:
\begin{align*}
\Gamma_\sigma\circ\Phi_0(X)=\sigma^{\frac{1}{2}}\sum_{j\in\mathcal{J}}K_j^*XK_j\sigma^{\frac{1}{2}}=\sum_{j\in\mathcal{J}}\e^{\omega_j}K_j^*\Gamma_\sigma(X)K_j=\Phi_*(\Gamma_\sigma(X))
\,.
\end{align*}	
Then, given any $\rho:=\Gamma_\sigma(X)\ge 0$,
\begin{align*}
D_{\operatorname{Lin}}(\Phi_*(\rho)\|\sigma)&\le D_{\operatorname{Lin}}D(\Phi_*(\rho)\|\mathcal{E}_*(X)/d_\cH)\\
&=D_{\operatorname{Lin}}(\Gamma_\sigma\circ \Phi_0(X)\|\Gamma_\sigma( \cE_0(X)))\\
&\le \|\sigma\|_\infty\,D_{\operatorname{Lin}}(\Phi_0(X)   \|  \cE_0(X)   )\\
&\le \|\sigma\|_\infty\,c(\Phi_0)D_{\operatorname{Lin}}(X\|\cE_0(X))\\
&  \le \|\sigma\|_\infty\,c(\Phi_0)D_{\operatorname{Lin}}(X\|1) \\
&= \|\sigma\|_\infty\,c(\Phi_0)\,D_{\operatorname{Lin}}(\Gamma_\sigma^{-1}(\rho)\|\Gamma_\sigma^{-1}(\sigma))\\
&\le \|\sigma\|_\infty\,\|\sigma^{-1}\|_\infty \,c(\Phi_0)\,D_{\operatorname{Lin}}(\rho\|\sigma)\,.
\end{align*}
where the first and fourth inequalities follow from Lemma \ref{chaintule}, whereas the second and last inequalities follows the same way as in Proposition \ref{propequivrelent}.
	\end{proof}

In \cite{miclo1997remarques}, Miclo devised a technique to estimate the SDPI constant of a doubly stochastic, primitive Markov chain in terms of the logarithmic Sobolev inequality of a corresponding Markov semigroup. This result was later generalized in \cite{RaginskyMaxima}.  An extension to the tracial quantum setting was recently provided in 
\cite{[MHFW15]}. Combining their result with our Proposition \ref{prop:sdpi}, we arrive at the following corollary. Given a unital, self-adjoint quantum Markov semigroup $(\e^{-t\cL_0})_{t\ge 0}$, its \textit{logarithmic Sobolev constant} is defined as
\begin{align*}
\alpha_2(\cL_0):=\inf_{X> 0}\,\frac{\frac{1}{d}\,\langle X,\cL_0(X)\rangle_{\operatorname{HS}}}{\tr\Big(   \frac{X^2}{d}\ln X^2\Big)-\tr\Big(  \frac{X^2}{d}\Big)\ln\tr\Big(\frac{X^2}{d}\Big)}\,.
\end{align*}
\begin{cor}
	Let $\Phi, \Phi_0$ be defined as in (\ref{prop:sdpi}) and assume that $\Phi_0$ is primitive. Then,
	\begin{align*}
		c(\Phi)\le \min\{ \|\sigma\|_\infty\,\|\sigma^{-1}\|_\infty\,(1-\alpha_2(\Phi_0^2-\id)),\,1\}\,.
		\end{align*}
	\end{cor}
\begin{proof}
	The bound follows from Proposition \ref{prop:sdpi} and Theorem 4.2 in \cite{[MHFW15]}.
	\end{proof}

\section{Preparation of mixed densities}\label{sec:stateprep}

 In this section our goal is to identify certain
 quantum Markov semigroups  which  can be used to prepare a given mixed state, and satisfies CLSI estimates. Clearly, the CLSI estimates will allow us to estimate the waiting time for a good approximation of the state.
 This complements  the results of \cite{kraus_preparation_2008} to the setting of mixed states.  For the rest of this section we will assume that
 \[ \sigma\lel \sum_{k=1}^m \si_k P_k \]
is a full-rank state on an $n$-dimensional Hilbert space, the $P_k$ are the projections onto eigenspaces.

\hz

\textbf{Eigenvalues of multiplicity one:}  Here we assume in addition that $\tr(P_k)=1$. Following \cite{Carlen20171810}, we know that the operators $A_j$ are eigenvectors of the modular operator $\Delta_\sigma$. In this particular case, these are given by the matrix units $E_{rs}:=|r\rangle\lan s|$ corresponding to a subset of edges  $E\subset \{1,...,m\}^2$. We may also choose
\[ A_{rs} \lel \chi_{rs} \,E_{rs}\pl~~~\text{ and }~~~ \pl A_{sr} \lel \chi_{rs} \,E_{sr}\,, \]
given some arbitrary constants $\chi_{rs}$
independent of the orientation of the edge.  The Bohr frequencies are given by $\Delta_{\si}(E_{rs})=\frac{\si_r}{\si_s}E_{rs}$ and hence $\om_{rs}=\ln \si_s-\ln \si_r$. Therefore the generator of the semigroup is given by
 \begin{equation} \label{linstate}
 \L_E(X) \lel \sum_{rs\in E}  |\chi_{rs}|^2 \kla  \Big(\frac{\si_r}{\si_s}\Big)^{1/2} (E_{ss}X -E_{sr}XE_{rs})+ \Big(\frac{\si_s}{\si_r}\Big)^{1/2}(XE_{rr}-E_{rs}XE_{sr}) \mer  \pl .
\end{equation}
Note that both terms are necessary for $\L_E$ to be the generator of a semigroup and we assume $\chi_{rs}\neq 0$. When the Bohr frequencies are all equal to $0$, the corresponding generator is denoted by $\cL_{E0}$.
\begin{defi} Let $E\subset \{(r,s)| 1\le r<s\le m\}$ be a subset of edges and $\tilde{E}=E\cup \{(s,r)|(r,s)\in E\}$. Then $E$ is said to be \emph{irreducible} if the graph with vertices $\{1,...,m\}$ and edges $\tilde{E}$ is connected.
\end{defi}

\begin{lemma}\label{erg2} $\L_E$ leaves the diagonal matrices $\ell_{\infty}^m\subset \Mz_m$ with respect to the basis $\{|r\rangle\}_r$ invariant.  Moreover, if $E$ is irreducible, then $\L_E$ is primitive.
\end{lemma}

\begin{proof}
	 For a diagonal matrix $X$ we have
	\[ E_{ss}X -E_{sr}XE_{rs} \lel (X_{ss}-X_{rr})E_{ss} \quad , \quad
	XE_{rr}-E_{rs}XE_{sr} \lel (X_{rr}-X_{ss})E_{rr} \pl .\]
	Thus $\L_E(\ell_{\infty}^m)\subset \ell_{\infty}^m$.
	
	Next, let us define $\delta_{rs}(x)=[E_{rs},x]$. Thanks to Lemma \eqref{ergodic} we note that $\cL_E(X)=0$ if and only if $[E_{rs},X]=0$ for all $(r,s)\in \tilde{E}$. Since the graph is irreducible, we can find a chain $(t,t_1),(t_1,t_2),...,(t_k,v)$ connecting any $t$ and $v$ and write
 \[ E_{tv}=E_{tt_1}\,E_{t_1t_2}\,\cdots\, E_{t_kv} \,.\]
Thus $X$ commutes with all matrix units and hence is a multiple of the identity.
 \qd

In the following we will assume that $\chi_{rs}=1$ for all $r,s$.

\begin{rem} \label{2exam}\rm It will be shown in a forthcomming work that for every irreducible $E$, $\cL_{E0}$ satisfies the CLSI. For the complete graph, i.e., when all edges are included, we see that
  \[\cL_{E0} (X)\,=\,2m\Big( X-\frac{1}{m}\sum_{k=1}^m X_{kk}\,\Id\Big) .\]
Therefore we deduce from \cite{BarEID17} that $\al_{\operatorname{CLSI}}(\L_{E0})\gl 2m$. For the cyclic graph
$E=\{(j,j+1)\}$ we know that $\al_{\operatorname{CLSI}}(\L_{E0})\gl \frac{c}{m^2}$ for some universal constant $c$.
\end{rem}

\begin{cor}\label{mone} Let $E\subset\{1,...,m\}^2$ be an irreducible graph with $\operatorname{CLSI}$ constant $\al_{\operatorname{CLSI}}(\cL_{E0})$. Assume further that $\si=\sum_{k=1}^m \si_k |k\rangle\lan k|$ is nondegenerate.  Then the $\operatorname{CLSI}$ constant of $\L_E$ satisfies
 \[ \al_{\operatorname{CLSI}}(\cL_{E0})\kl \max_{kl} \frac{\si_k}{\si_l}\, \max_{(rs)\in \tilde{E}} \Big(\frac{\si_s}{\si_r}\Big)^{1/2}\, \al_{\operatorname{CLSI}}(\L_E) \pl .\]
\end{cor}

\begin{proof} This follows directly from Theorem \ref{theocomp} and Lemma \ref{erg2}.
\qd

\begin{rem} {\rm For Lindblad operators with coefficients $\chi_{rs}$, we obtain a similar result. Note however, we should expect a normalization requirement due to the geometry of the H\"ormander system introduced in \cite{gao2018fisher}.}
\end{rem}

\begin{rem} a) For simple multiplicities, the analogy with graph Laplacians on a finite set of vertices  goes very far. For this comparison we denote by $\mu$ the measure with probabilities $\mu(\{k\})=\si_k$.

\begin{enumerate}
\item[i)] The space $L_2(\mu)$ sits as a diagonal in the space $L_2(\Mz_m,\si)$ with inner product $\langle X,Y\rangle_{\si}=\tr(\Gamma_{\si}(X)Y)$.
\item[ii)] On $L_2(\mu)$ we may consider the derivations $\delta_{rs}:\ell_{\infty}^m\to \ell_{\infty}^2$, $\delta_{rs}(f)=(f(r),-f(s))$. The graph Laplacian is given by
   \[ A_{E} \lel \sum_{rs\in E} \delta_{rs}^*\delta_{rs}\]
Moreover, $\hat{\delta}_{rs}:\Mz_m\to \Mz_m$, $\hat{\delta}_{rs}(X)=[E_{rs},X]$ extends these derivations to $L_2(\Mz_m,\si)$ and
 \[ \cL_{E0} \lel \sum_{rs} \hat{\delta}_{rs}^*\hat{\delta}_{rs} \]
extends the operator $A_E$, i.e. $\L_{E0}|_{L_2(\mu)}=A_E$. In particular, $e^{-t\L_{E0}}(f)=e^{-tA_E}(f)$ for diagonal operators $f$.
\item[iii)] According to \cite[section 5]{Carlen20171810} the Bohr frequencies in \eqref{linstate} appear naturally by duality with respect to the inner product given by $\si$. Finally, transferring the semigroup to $\Mz_m$ or $(\Mz_m)_*$ then induces $\L_E$ or $\L_{E*}$.
\end{enumerate}

b) The easiest choice of a stabilizing semigroup from \cite{MSFW}, given by
 \[ \P_{t*}(\rho) \lel e^{-t}\rho+(1-e^{-t})\si \,,\]
corresponds to taking the complete graph, certainly a very convenient choice, which requires to `average over many edges'. It is easy to see that $\al_{\operatorname{CLSI}}(\L_{*})\ge1$ for the generator $\cP_{t*}=e^{-t\L_*}$ (see e.g. \cite{BarEID17}). Thus for more practical applications it remains to be seen which Lindblad operators can be `implemented' and then aim for an estimate of the corresponding $\operatorname{CLSI}$ constant of $\cL_0$, which is independent of $\si$.
\end{rem}
\hz
\textbf{Eigenvalues of larger multiplicity and locality:}  We will now extend Corollary \eqref{mone} to the general case following the same procedure. Recalling that $\sigma=\sum_{k}\,\sigma_kP_k$, let us define the subspaces $\cH_k=P_k\cH$ and write
 \[ \{1,...,n\} \lel \bigcup_k I_k \pl \,,\]
where each subset $I_k$ corresponds to the eigenspace $\cH_k$. We will also assume that the eigenvalues $\si_k$ of $\sigma$ are defined in decreasing order.  Then we choose edges $E\subset \{1,...,n\}^2$ and consider
 \begin{equation} \label{linstate2}
 \L_E(X) \lel \sum_{(r,s)\in E}   \Big( e^{-\om_{rs}/2}(E_{ss}X-E_{sr}XE_{rs})+
 e^{\om_{rs}/2}(XE_{rr} -E_{rs}XE_{sr})
 \Big) \pl ,
 \end{equation}
where
 \[ e^{-\om_{rs}/2}  \lel \begin{cases} 0& \mbox{if there exists a $k$ such that} \pl  r, s\in I_k \\
  \Big(\frac{\si_k}{\si_j}\Big)^{1/2} & \mbox{if there exists $k\neq j$ such that} \pl  r\in I_k, s\in I_j \pl .
  \end{cases}   \]

\begin{cor} Let $E$ be an irreducible  graph and $\sigma=\sum_k \si_k P_k$. Then the semigroup $\L_E$ satisfies $\operatorname{CLSI}$ and
 \[ \al_{\operatorname{CLSI}}(\cL_{E0}) \kl \max_{k,l} \frac{\si_k}{\si_l} \max_{I_k\times  I_j\cap \tilde{E}\neq 0} \Big(\frac{\si_k}{\si_l}\Big)^{1/2} \pl \al_{\operatorname{CLSI}}(\L_E) \pl . \]
\end{cor}

Note here that the graph structure is by no means necessary. In particular, we could use any nice set of generators to produce primitive semigroups in $\mathcal{B}(\cH_k)$. Moreover, once this is achieved we just need sufficiently many links $A_{kj}\in \mathcal{B}(\cH_k,\cH_j)$ to guarantee primitivity of $\L_0$.

\hz

We assume now that $\cH^{\ten d}$ is a $d$-partite system, and denote by $\cH_k$ the eigensubspaces of $\sigma\in\mathcal{D}(\cH^{\otimes d})$. We say that a subspace $\cK\subset \cH^{\ten d}$ is $l$-local if there exists a permutation $\pi:\{1,...,d\}\to \{1,...,d\}$ and a projection $Q\in \mathcal{B}(\cH^{\ten l})$ such that $P_\cK=\Si_{\pi}^{-1}(Q\ten 1_{\cH^{\ten (d-l)}})\Si_{\pi}$, where $\Si_{\pi}$ is obtained by permuting registers: $\Si_{\pi}(h_1\ten\cdots \ten h_d)=h_{\pi(1)}\ten \cdots h_{\pi(d)}$. Similarly we say that an operator $A$ is $l$-local if $A\cong B\ten 1_{\cH^{\ten {(d-l)}}}$ holds up to a permutation of registers. The same definition holds for superoperators.

\begin{lemma} Assume that for all $k$ the subspaces $\cH_k$  are $l$-local, and that for $I_k\times I_j\cap E$ the space $\cH_k+\cH_j$ is $l$-local. Then $\L_{E*}$ is $l$-local.
\end{lemma}

\begin{proof} We recall that
 \[ \L_{E*}(\rho)  \lel \sum_{(r,s)\in E}    e^{-\om_{rs}/2}(A_{rs}^*A_{rs}\rho-A_{rs}\rho A_{rs}^*)+
 e^{+\om_{rs}/2}(\rho A_{rs}A_{rs}^* -A_{rs}^*\rho A_{rs}) \pl . \]
Here we may replace the matrix units by $E_{rs}\ten 1$ for $r\in I_{l}$ and $s\in I_{l}$ up to permutation. Similarly, we can stabilize the space $1\ten \cH^{\otimes d-l}$ using a Laplacian
 \[ \L \lel \sum_{i=1}^{d-l} \Id\ten \cdots \otimes\underbrace{\cL_H}_{\mbox{$i$-th position}}\ten \cdots\otimes \Id \]
for a primitive Laplacian   on $\mathcal{B}(\cH)$ given by commutators. Then $\L_{E0}$ and $\L_{E*}$ will only use local operators $A_{rs}$ and second order differential operators of the form $[A_{rs}\rho,A_{rs}^*]$.
\qd

\begin{rem} 1) The semigroups $(\cP_t=e^{-t\L_0})_{t\ge 0}$ for unital $\L_0$ can be obtained in the form $\cP_t(X)=\ez[ \,U_t^*XU_t\,]$ for random unitaries.
The approximation of $e^{-t\L_{E*}}\approx 1-t\L_{E*}$ with unitary operations  will be investigated in a forthcomming publication.

2) Nevertheless in the local situation operators $A_{rs}$ do not really depend on the state per se, just on its eigen-projections. The Bohr-frequencies however, drive the QMS in the specified state $\si$.
\end{rem}

\hz
\textbf{State preparation using history states:} We will modify a construction from \cite{verstraete2009quantum}. Our starting point is a faithful density $\rho_0\in \cD_+(\cH^{\otimes d})$ and a Lindbladian $\L_{log}(\rho_0)=0$. A suitable choice is $\rho_0=\rho_{00}^{\otimes d}$, $\rho_{00}\in\cD_+(\cH)$, and $\L_{log}=\sum_{j=1}^d \pi_j(\L_{00})$ where $\pi_j$ refers to the $j$-th register. For a given discrete set of times parameters $0\le s\le S$, we fix unitaries $V_s=U_{s}\cdots U_{1}$ so that $U_{s}$ are local unitaries. Our goal is to prepare the state
  \[ \rho_{T} \lel U_T\rho_0U_T^* \pl .\]
For mathematical reasons, it is easier to start with $U_s=1$. Then we may use the nearest neighbour Linbladian
 \[ \L_{tim}(X) \lel \sum_{s=0}^T W_s^*W_sX+XW_sW_s^*+w_s^*w_sX+Xw_sw_s^* -2W_s^*XW_s-2w_s^*Xw_s \pl, \]
where $W_s=|s\rangle\lan s+1|$ and $w_s=|s+1\rangle\lan s|$ are the generator of the quantum version of the cyclic graph. Therefore, we have
 \[ \kappa\pl:=\pl \al_{\operatorname{CLSI}}(\L_{log}\ten id+id\ten \L_{tim})\gl \max\{
\al_{\operatorname{CLSI}}(\L_{log}),\frac{c_0}{T^2}\} \pl .\]
In the tensor product situation we may assume  $\al_{\operatorname{CLSI}}(\sum_j\pi_j(\L_{00}))\gl \al_{\operatorname{CLSI}}(\L_{00})\gl \al_0$. In order to prepare the actual state $\rho_T$ we use the unitary $U\lel \sum_{s} U_s\ten |s\rangle\lan s|$ and combined Lindlbadian
 \[ \L_U \lel  ad_{U}(\L_{log}\ten id +id\ten \L_{tim})ad_{U^*} \]
where $ad_U(X)=UXU^*$.

\begin{cor} Let $X$ be a density in $\mathcal{T}_1(\cH^{\otimes n})$. With probability $1/(T+1)$ the density  $X_T=T(id\ten |T\rangle\lan T|(e^{-s\L_U}(X \ten \frac{id}{T})(id\ten |T\rangle\lan T|)$
 \[ D_{\Lin}(X_T\|\rho_T)\kl (T+1) e^{-\kappa s} D(X\|\rho_0) \pl .\]
\end{cor}

Here $\phi_T:\mathcal{T}_1(\ell_2^{T+1})\to \cz$,
$\phi_T((x_{rs}))=x_{TT}$ is a trace reducing map and $\phi_T(e^{-s\L_U}(\rho_0\ten \frac{id}{T}))=\rho_T$, and hence we may apply data processing inequality, see Lemma 3.3 i).
 Note that semigroup $e^{-s\L_U}$ preserves the subalgebra of time diagonal operators. This means that any time $s$, the time can be measured precisely, as a classical parameter. Thanks to the product form, we also know that $\L_{tim}$ preserves the the measure $\mu(s)=\frac{1}{S}$, and hence $D_{\Lin}(X_T|\rho_T)=D(X_T|\rho_T)$. However, the same statement remains true, using $D_{\Lin}$, if we decide to use $X_0\lel X\ten |0\rangle\lan 0|$ as an input for the combined time-logical Lindbladian. Thanks Pinsker's inequality we know then
  \[ \|\|X_T\|_1-\frac{1}{T+1}\|_1\kl \|e^{-s\L}(X_0)-\lim_{s\to \infty}e^{-s\L}(X_0)\|_1
  \kl 2 \sqrt{e^{-ks}D(X_0|\rho_0)} \]
therefore we control both the original relative entropy and the trace.
\hz

The Lindblad operators in $\L_U$ are all local, at least as local as the $W_s$, $w_s$'s and  $ad_{U_s}(V_j)$, where $V_j$ runs through  the Lindblad generators of $\L_{00}$.

\hz

\begin{rem} {\rm Following the model in \cite{verstraete2009quantum} we could also consider the Lindbladian $\L=\L_{log}+L_{tim}$ defined as follows: Let $V_i$ be the logical Kraus operators for $\L_0$ with frequences $\om_j$ and define $A_i \lel V_i\ten |0\rangle \langle 0|$ and
 \[ W_t \lel U_{t+1}\ten |t+1\rangle \langle t+ U_{t+1}^* \ten |r\rangle \langle t+1| \]
Then the new Linbladian would be given by
 \begin{align*}
   \L(X) &=  -\sum_{j} (e^{-\om_j/2}[A_i^*[X,A_i]+e^{\om_j/2}[A_j,X]A_j^*)   \\
    &\pll + \sum_{t} W_t^*W_tX+XW_tW_t^*+W_tW_t^*X+XW_t^*W_t^*-2W_t^*XW_t-2W_t^*XW_t \pl .
    \end{align*}
The advantage of this form is that specific information on the gate only interferes in the time Laplacian.     Unfortunately, at the time of this writing, we have no CLSI-bound in the self-adjoint case $\om_j=0$, and hence no concrete bound on relative entropy.}
\end{rem}

We may easily improve on the factor $T$ by working with more classical resources by using $\Om=\{0,...,T\}^m$,
 \[ U_{s_1,...,s_m}\lel U_{\max_j s_j} \quad \mbox{and}\quad U\lel \sum_{\om} U_{\om}\ten |\om \rangle \langle \om|  \]
and
 \[ \L_U \lel ad_{U}(\L_0\ten id+id\ten (\sum_{j=1}^m \pi_j(\L_T))ad_{U^*} \pl .\]
Then standard probabilistic method is to use the classical stopping time $s_m(\om)=\inf\{k| s_k=T\}$. Note that error probability  for failure
 \[ \eps_m(T)\lel {\rm Prob}(s=\infty) \lel (1-1/T)^m \]
goes to $0$ exponentially fast. We use the notation $e_m$ for the projection onto successful events. Then we may use the trace reducing map
 \[ \Phi_m(Y) \lel
 tr\ten id((e_m\ten 1)Y(e_m\ten 1)) \pl .\]

\begin{cor} Let $m\in \nz$ and $X_0\in \mathcal{T}_1(\cH^{\otimes d})$. Let $X_{s,m}=(1-\eps_m)^{-1}\Phi_m((e^{-s\L_U})(ad_U(X\ten |0\rangle\langle 0|)))$. Then
 \[ D_{\Lin}(X_{s,m}\|\rho_T)\kl (1-\eps_m)^{-1} e^{-\kappa s} D(X\|\rho_0) \pl.
 \]
\end{cor}

This means by renormalizing  $\hat{X}_{s,m}=\frac{X_{s,m}}{tr(X_{s,m})}$, we control the relative entropy
 \[ D(\hat{X}_{s,m}\|\rho_T)\kl e^{-\frac{\kappa}{2}s} D(X_0\|\rho_0) \]
where $m\gl T \max\{1,2|\ln(\kappa s)|\}$. This provides example of dissipative state preparation for the output the gate $U_T$ applied to $\rho_0$. The stopping time procedure, also allows for a recursive algorithm to produce $U^{k+1}_{\om}\lel \begin{cases} U^k_{\om}& s_{k+1}\le \max_{j\le k} s_j \pl \\
V_{s_{k+1}}\cdots V_{s_k(\om)}U^k_{\om} & s_{k+1}>\max_{j\le k} s_j\pl  \end{cases}$ so that $U \lel U^m_{\om}$.

\section{Decay Toward Thermal Equilibrium} \label{sec:therm}
A system in contact with a heat bath at fixed temperature will decay toward a Gibbs state given by
\begin{equation} \label{eq:thermalstate}
\sigma_\beta = \frac{e^{-\beta H}}{Z_\beta}=\frac{1}{Z_\beta} \sum_{k=1}^m e^{-\beta E_k} \ketbra{k} = \frac{1}{Z_\beta} \sum_{{ E\in \spec{(H)}}} e^{-\beta {E}} \sum_{k : E_k = E} \ketbra{k} \pl,
\end{equation}
where $H$ is the corresponding Hamiltonian, $\beta$ is the unitless inverse temperature, and energies are indexed in increasing order. Here the last expression explicitly separates the sum over possibly degenerate energies. The partition function $Z_\beta = \sum_k \exp(-\beta E_k)$ normalizes the probabilities.

As an important example of thermal state decay, usual decoherence processes in quantum information experiments do not necessarily decay to white noise. The commonly reported ``$T_1$" or longitudinal coherence time theoretically assumes decay toward a ground state \cite{chuang_lecture_2003}. While our usual formalism does not apply to pure states, we may consider the simple quantum Markov semigroup,
\begin{equation*}
P^{\text{relax}}_{t*}(\rho) = e^{- t / T_1} \rho + (1 - e^{- t / T_1}) \ketbra{0} \pl,
\end{equation*}
which models a process containing only irreversible transitions to a fixed pure state. By convexity of relative entropy,
\begin{equation*}
D((P^{\text{relax}}_{t*} \otimes \id_B)(\rho) \p \| \p  \ketbra{0} \otimes \rho_B) \leq e^{-t / T_1} D(\rho \p \| \p \ketbra{0} \otimes \rho_B)
\end{equation*}
for a bipartite state $\rho_{AB}$. As expected, this process has $1/T_1$-CLSI. Note however that the relative entropy of this state is usually infinite, so this comparison is of very limited practical value. In practice, we often expect matter qubits to decay toward a low-temperature thermal state. On the preparation side, we may wish to heat or cool a system to a desired temperature. Directly following Theorem \ref{theocomp},
\begin{cor} \label{cor:therm1}
Let $\omega_\beta$ be a thermal state as in Equation \eqref{eq:thermalstate} and the fixed point of QMS $(\cP_t)_{t\ge 0}$ generated by Lindbladian $\L$ with $m$ energy levels $E_1< E_2<...<E_m$. Let $\L_0$ be the corresponding self-adjoint Lindbladian. Then
\begin{equation}
\alpha_{\operatorname{CLSI}}(\L_0) \leq e^{\beta E_m} \max_j e^{\beta (E_{j+1} - E_j) / 2} \,\alpha_{\operatorname{CLSI}}(\L)\pl.
\end{equation}
\end{cor}
\noindent The completely mixed state is equivalent to the infinite temperature Gibbs state $\sigma_0$, so we might think of the CLSI constant comparison as perturbing the infinite-temperature limit. With a finite maximum energy and for temperatures substantially above that scale, this CLSI constant is close to that for decay toward complete mixture.

In general, our CLSI constant estimate depends exponentially on the largest energy scale and becomes trivial with an infinite spectrum. This appears to be not a flaw of the estimate, but a property of relative entropy. When high-energy elements of the Gibbs state are close to 0, even small fluctuations into the high-energy regime can result in enormous relative entropy. In thermodynamics, a usual solution to relative entropy blowup on rare states is to work with smoothed relative entropies \cite{faist_fundamental_2018}, which discount contributions from highly unlikely configurations. While it is beyond the scope of this paper to fully formulate log Sobolev inequalities for smoothed relative entropy, we may nonetheless consider an analogous approach for states that rarely occur. 

\begin{rem}
A simple strategy is to replace the Gibbs state by
\begin{equation}
\tilde{\sigma} = \frac{1}{\tilde{Z}_\beta} \Big ( \sum_{k=1}^{l-1} e^{-\beta E_k} \ketbra{k} + \sum_{k=l}^m e^{-\beta E_l} \ketbra{k} \Big ) \pl,
\end{equation}
for which the associated Lindbladian $\tilde{\cL}$ has
\begin{equation*}
\alpha_{\operatorname{CLSI}}(\L_0) \leq e^{\beta E_l} \max_{j \leq l} e^{\beta (E_{j+1} - E_j) / 2} \,\alpha_{\operatorname{CLSI}}(\tilde{\L})\pl.
\end{equation*}
Physically, this is equivalent to artificially compressing high energies to a single, degenerate level. We do not claim that this accurately represents the high-energy parts of the thermal state or decay of states with substantial support above $E_l$. Rather, $\tilde{\L}$ is an example of a Lindbladian with the same transitions as $\L$ and similar low-energy behavior at short timescales. It hence naturally has the same locality properties.

The distance $\| \sigma - \tilde{\sigma}\|_1$ increases with the value of $E_{j+1}$ and higher levels. We can overestimate it by assuming $E_{l+1} = \infty$, as though $\sigma$ had no support in the high-energy space. Let $\sigma_0 = \sum_{k=1}^l e^{-\beta E_k}$, and $\sigma_1 = \tilde{\sigma} - \sigma_0$. We can easily check that $0 \leq \tilde{Z}_\beta - Z_\beta \leq (m-l) e^{- \beta E_k}$, and similarly, $\| \tilde{\sigma} - \sigma_0 \| \leq (m-l) e^{- \beta E_k}$. Hence
\begin{equation*}
\|\tilde{\sigma}_\beta - \sigma_\beta\|_1 \leq \frac{m-l}{\tilde{Z}_\beta} e^{- \beta E_k} \Big | \frac{1}{\tilde{Z}_\beta} - 1 \Big | + O \big ( (m-l)^2 e^{- 2 \beta E_k} \big ) \pl,
\end{equation*}
which decreases exponentially with $E_l$.
\end{rem}
This simple solution lacks two desirable properties. First, we have no quantitative notion of this $\tilde{\L}$ being close to the original physical process, other than that it involves the same jump operators. Second, the rate of convergence still depends linearly on the number of energies deemed high, which prevents this trick from approximating problems with infinite maximum energy, such as in the quantum harmonic oscillator. We can partially mitigate these problems by considering an alternate model. By the semigroup property and for any $m \in \NN$,
\begin{equation*}
P_{t*}(\rho) = P_{t/m*}^{m}(\rho) \pl.
\end{equation*}
For large $m$, we may approximate
\begin{equation*}
P_{t/m*}(\rho) = \rho - \frac{t}{m}\L(\rho) + O \Big (\frac{t^2}{m^2}\Big ) \pl.
\end{equation*}
Let us define the measurement $M_E: \mathbb{B}(\cH) \rightarrow \mathbb{B}(\cH) \otimes l_\infty^2$ by
\[ M_E(\rho) = P_{E \leq E_0} \rho P_{E \leq E_0} \otimes \ketbra{0} + P_{E > E_0} \rho P_{E > E_0} \otimes \ketbra{1} \pl, \]
where $P_{E\le E_0}$ denotes the projector onto the eigensubspace of $H$ of eigenvalues less than $E_0$, $P_{E>E_0}=\Id-P_{E\le E_0}$, and $l_\infty^2$ holds a classical, binary flag. We then construct
\[ \tilde{P}_{t^*}^m : \cH \rightarrow \cH \otimes (l_\infty^2)^{\otimes (m + 1)}, \pl \tilde{P}_{t*}^m(\rho) = (M_E \circ P_{t/m*})^{m} \circ M_E(\rho) \pl. \]
Let $p_{low}(\L, \rho, t, m)$ be the probability that none of the $m+1$ classical flags in $\tilde{P}_{t*}^m$ are 1. Let $p_{low}(\L, \rho, t) \equiv \lim_{m \rightarrow \infty} p_{low}(\L, \rho, t, m)$, and $\tilde{P}_{t^*}^\infty \equiv \lim_{m \rightarrow \infty} \tilde{P}_{t^*}^m$. In the limit as $m \rightarrow \infty$, we may assume that only one jump if any occurs between any measurement pair. Hence $p_{low}(\L, \rho, t)$ is the probability that the state never transitions through the high-energy subspace.

To obtain an approximation result, we need two assumptions: first, that $p_{low}$ is non-zero, and second, that any passage through the high-energy space would be recorded in some environment even if not available to the experimenter. The latter implies that we could add the measurements without disturbing the system, so we may consider $P_{t*}(\rho)$ to be the channel obtained by applying $\tilde{P}_{t^*}^\infty$ and tracing out the high-energy flags.
\begin{prop} \label{prop:lowenergy}
Let a density $\rho$ and primitive  Lindbladian $\L$ with detailed balance and thermal fixpoint state $\sigma$ be given. Let $E_0$ be a fixed energy cutoff defining the low-energy subspace. Let $\tilde{\sigma} = \frac{1}{\tilde{Z_\sigma}} P_{E \leq E_0} P_{t/m*}(\sigma) P_{E \leq E_0}$, where $\tilde{Z}_\sigma$ is a factor assuring the state is normalized. Let $\tilde{P}^\infty_{t*}$ be defined with respect to the energy cutoff $E_0$. Assuming any passage through the high-energy space in which $E > E_0$ would be recorded in the environment and that any jump operators through the low-energy subspace are within the subalgebra given by $P_{E \leq E_0} \mathbb{B}(H) P_{E \leq E_0}$,
\[D(\tilde{P}^\infty_{t*}(\rho) \| \tilde{\sigma}) \leq e^{-\alpha t} D(\tilde{\rho} \| \tilde{\sigma}) \]
with probability at least $p_{low}(\L, \rho, t)$, where
\begin{equation*}
\alpha_{\operatorname{CLSI}}(\L_0) \leq e^{\beta E_0} \max_{E_j \leq E_0} e^{\beta (E_{j+1} - E_j) / 2} \,\alpha\pl,
\end{equation*}
and $\tilde{\rho} = P_{E \leq E_0} \rho P_{E \leq E_0} / \tilde{Z}_\rho$.
\end{prop}
\begin{proof}
The assumption that any jump through the high-energy subspace would be recorded in the environment means that 
\begin{equation*}
P_{t/m*}(\eta) = \tr_{l_2^2}((M_E \circ P_{t/m*})(\eta)) = P_{E \leq E_0} P_{t/m*}(\eta) P_{E \leq E_0} + P_{E > E_0} P_{t/m*}(\eta) P_{E > E_0}
\end{equation*}
for any $m$ and density $\eta$ that starts entirely in the low-energy subspace. If we assume that the environment recorded no such jump, this is equivalent to projecting onto the low-energy subspace every time. Here
\[P_{E \leq E_0} P_{t/m*}(\eta) P_{E \leq E_0} = \eta - \frac{t}{m} \sum_{j : E_j \leq E_0} \Big(e^{-\om_j/2}[A_j\eta,A_j^*]+e^{\om_j/2}[A_j^*,\eta A_j] \Big ) + O \Big ( \frac{t^2}{m^2} \Big ) \pl. \]
Resumming this in the limit as $m \rightarrow \infty$, we arrive at an effective Lindbladian that includes only the low-energy jumps.
\end{proof}
\begin{rem}
The assumption of jump operators being contained in the low-energy subalgebra in Proposition \ref{prop:lowenergy} is satisfied by jump operators of the form $A_j = \ket{r_j}\bra{s_j}$, such that $E_j \leq E_0$. Equivalently, $r_j, s_j \leq d_0$, where $d_0^2$ is the dimension of the low-energy subspace, and states are ordered by increasing energy.
\end{rem}

Proposition \ref{prop:lowenergy} confirms the intuition that low-temperature processes may see little to no effect from high-energy transitions. E.g. a cold atom experiment is unlikely to encounter consequences of nuclear physics. Were one to have a way of measuring whether or not the state transitioned into the high-energy subspace, one could use this with post-selection to probabalistically prepare exact copies of the low-energy projection of the fixpoint state.

There are likely physical situations for which the assumption of no transitions to the high-energy space could be replaced by one that they are sufficiently rare as not to substantially disrupt the process. As the purpose of this manuscript is not to study specific physical systems in detail, we do not study these cases here. For similar reasons, we will not attempt to compute $p_{low}(\L, \rho, t, m)$ for specific Lindbladians. In many physical study systems, it should be clear how to compute or at least estimate this probability.

\bibliographystyle{alpha}
\bibliography{library}
\end{document}